\documentclass[10pt,english]{article}
\usepackage[utf8]{inputenc}
\usepackage[left=3cm,right=3cm,top=2cm,bottom=2.5cm]{geometry}

\usepackage{amsmath}
\usepackage{amssymb}
\usepackage{mathtools}
\usepackage{amsthm}
\usepackage{natbib}
\usepackage{hyperref}
\usepackage{cleveref}

\usepackage{algorithm}
\usepackage{algorithmic}
\usepackage{booktabs}

\def\jie#1{\textcolor{red}{Jie: #1}}

\usepackage{physics}
\usepackage[T1]{fontenc}
\usepackage{caption}
\usepackage{graphicx}
\usepackage{multirow}
\usepackage{makecell}
\DeclareMathOperator*{\Ex}{\mathbb{E}}
\DeclareMathOperator*{\Prob}{\mathbb{P}}

\newcommand{\algo}{$\mathsf{EgalAlg}$}
\newcommand{\appalgo}{$\mathsf{EgalAlg}^{(\rm appx)}$}
\newcommand{\y}{y}
\newcommand{\hy}{\hat{y}}
\newcommand{\ty}{\tilde{y}}
\newcommand{\ihy}{\hy^{\rm indv}}
\newcommand{\ghy}{\hy^{\rm gr}}
\newcommand{\z}{z}
\newcommand{\vz}{\mathbf{z}}
\newcommand{\err}{\mathsf{err}}
\newcommand{\cor}{\mathsf{tru}}
\newcommand{\N}{\mathcal{N}}
\newcommand{\lW}{\bar{W}}
\newcommand{\hW}{\widehat{\bar{W}}}
\newcommand{\imp}{\varphi}
\newcommand{\znew}{\z_{\rm new}}
\newcommand{\Gagg}{\mathcal{G}\strut^{({\rm agg})}}
\newcommand{\Gegal}{\mathcal{G}\strut^{({\rm egal})}}
\newcommand{\B}{\mathcal{Z}}
\newcommand{\Bnew}{\mathcal{\B}_{\rm new}}
\newcommand{\OPT}[1]{{\rm OPT^{({\rm #1})}}}
\newcommand{\Gain}{\Delta{\mathcal G}}
\newcommand{\val}{\Psi}
\newcommand{\Greedy}{{\mathsf{gr}}}

\newcommand{\R}{{\mathcal{R}}}
\newcommand{\Bl}{{\mathcal{B}}}
\newcommand{\W}{{\mathcal{W}}}
\newcommand{\T}{{\mathcal{T}}}
\newcommand{\F}{{\mathcal{F}}}
\newcommand{\C}{{\mathcal{C}}}
\newcommand{\bC}{{\bar{\mathcal{C}}}}

\newcommand{\cvratio}{{\rm Acc}}

\newcommand{\shahrzad}[1]{\textcolor{blue!60!black}{\textsc{Shahrzad:} \emph{#1}}}
\newcommand{\cheng}[1]{\textcolor{cyan}{\textsc{cheng:} \emph{#1}}}

\newtheorem{problem}{Problem}

\theoremstyle{plain}
\newtheorem{theorem}{Theorem}[section]
\newtheorem{proposition}[theorem]{Proposition}
\newtheorem{lemma}[theorem]{Lemma}
\newtheorem{corollary}[theorem]{Corollary}
\theoremstyle{definition}
\newtheorem{definition}[theorem]{Definition}
\newtheorem{assumption}[theorem]{Assumption}
\theoremstyle{remark}
\newtheorem{remark}[theorem]{Remark}

\author{
  Shahrzad Haddadan \\
 Rutgers Business School\\
\texttt{shaddadan@business.rutgers.edu}\\
  \and
  Cheng Xin\\
 Department of Computer Science\\ Rutgers University\\
\texttt{cx122@rutgers.edu}\\
\and 
  Jie Gao\\
 Department of Computer Science\\ Rutgers University\\
 \texttt{ jg1555@rutgers.edu}
}

\title{Optimally Improving Cooperative Learning in a Social Setting\footnote{This paper is going to appear in ICML 2024}
}

\begin{document}

\maketitle








\begin{abstract}
We consider a cooperative learning scenario
where a collection of networked agents with individually owned classifiers dynamically update their predictions, for the same classification task, through communication or observations of each other's predictions.
Clearly if highly influential vertices use erroneous classifiers, there will be a negative effect on the accuracy of all the agents in the network. We ask the following question:
how can we optimally fix the prediction of a few
classifiers so as maximize the overall accuracy in the
entire network.
 To this end we consider an aggregate and an egalitarian objective function. 
 We show  a polynomial time algorithm for optimizing the aggregate objective function, and show that optimizing the egalitarian objective function is NP-hard. Furthermore, we develop approximation algorithms for the egalitarian improvement. The performance of all of our algorithms are guaranteed by mathematical analysis and backed by experiments on synthetic and real data. 
\end{abstract}

\section{Introduction}\label{sec:intro}

With the breakthrough of AI technologies and the availability of big data, we are witnessing a flourish of AI models that are owned by different entities and trained by using private or proprietary data, even for generic purposes such as voice recognition, natural language processing or image segmentation. These models do not stay in isolation. There is a natural opportunity for these models to interact with each other and collectively improve their performance.

One of the concrete application scenarios is in cybersecurity, where security agents in a network collectively detect anomalous traffic patterns that are potentially associated with cyber attacks. These security agents may be affiliated with different entities in the network. They may or may not be able to directly share their collected data due to privacy or other logistic reasons but can share their beliefs on whether the network is under attack or not, and if so, which type of attack. In such a scenario, improving the model of one agent by collecting more data, recruiting domain experts to provide high quality labels of observation data, or retraining using a more powerful model can help to improve the quality of prediction locally. As these security agents stay on a network, they naturally have the opportunity to exchange their assessments. The quality improvement at one agent spills to other agents in the network.  

Another application scenario is in online social networks. With generative AI, it is now a lot easier to create fake contents such as images and videos. 
Consider an online social network in which agents share pictures and comment on them, e.g., Instagram. Assume  that some of these pictures may be AI generated. While everyone can have his or her opinion on whether a wide-spreading picture is real or fake, the participants who are more skilled in image generation or who have access to powerful models can help the rest of the network to discern real ones and fake ones. 

It is not new to utilize data and models in a distributed setting.
With the explosion of personal data from smartphones and wearable devices and the increasing awareness of data privacy, decentralized learning, as opposed to users sharing their data to a central enterprise, has gained popularity in recent years. Federated learning and gossip learning are such examples. However in this scenarios, the decentralized agents are still tightly coupled -- they use the same model architecture and sometimes exchange gradient or model parameters directly. 
In our work, we consider a loosely coupled, cooperative setting where agents share a general task requirement and they select individual model parameters and architecture, train their models on their private data. Such agent autonomy is a necessity when the agents are not affiliated with the same ownership.
In addition, we consider the agents sharing their {\it predictions} with other agents preferably those within proximity or with established trust relationships. We call this model {\it cooperative learning in a social setting}. Essentially each agent has her own view of the world and through exchanging 
predictions with others, collectively we hope to improve the accuracy for all agents.

When only predictions from other agents are shared, it is up to the individual agent to decide how to update their models. On a first-order basis, we can assume for now that the outcome of such information sharing can be approximated by a weighted linear combination of the agent $u$'s own assessment and the predictions from other neighbors. The weights can be either fixed or time-varying, e.g., based on the trust level of neighbors. This leads to two natural models, 
\citet{DeGroot1974-ed} and 
\citet{Friedkin1990-wl}, originally proposed in modeling opinion dynamics in social network. In DeGroot model, each agent's prediction is a simple weighted linear combination of neighbors' predictions. When the weight coefficients are fixed (or when the updates are frequently enough for time-varying weights), all agents converge to global consensus. In FJ model, each agent incorporates in the update step a vector of personal assessment (which can be guided by the difference of local data distribution). The model still converges, and each agent arrives at possibly different predictions, reflecting adaptation to individual local data distributions. 

The generic framework captures how a group of networked agents build their respective models collectively. When new training data is introduced to one agent in the system, the other agents indirectly receive the benefit of it. Therefore it is an interesting question to analyze the collective benefit and also ask the optimization question of \emph{where to inject new training data to maximize the collective benefit}. This is the research question we focus on in this paper.

\subsection{Related Work}

\paragraph{Decentralized learning}
Training a single high-quality global model using decentralized data and computation has been studied extensively in decentralized optimization (e.g., for kernel methods~\cite{Colin2016-eh}, PCA~\cite{Fellus2015-re}, stochastic gradient descent~\cite{Blot2016-qy}, multi-armed bandit~\cite{Lazarsfeld2023-hc} and generalized linear models~\cite{He2018-sl}). 
Federated learning~\cite{Konecny2015-qn, Konecny2016-qy,Konecny2016-ie, McMahan2017-em,Krishnan2024-wn} uses a client-server architecture and considers multiple local models, trained using respective local data, with model parameters aggregated and shared through a central global model. Gossip learning~\cite{Ormandi2013-gk, He2018-sl, Hegedus2019-fa,Hegedus2016-ud,Giaretta2019-yf}, on the other hand, does not assume a central node. Instead, each node updates its own local parameters via training and then its updated parameters are shared by information exchange with other nodes in the network. This setting removes the single point of failure in the system and thus is more robust and scalable, without compromised performance~\cite{Hegedus2016-ud,Hegedus2021-ac}. 
These gossip learning protocols consider the exchange of local models (or crucial parameters such as local gradients) directly. This requires that all agents participating in gossip learning use the same type of models, which is a limitation. It is also known that models or gradients can reveal knowledge of the training data (thus raising concerns to data privacy, see a recent survey here~\cite{Zhang2022-wr}).

\vspace{-0.35cm}
\paragraph{Social learning} The study of learning and decision making  in a social network 
has been studied for longer than a decade. In these works, agents predict the status of the world, and based on their prediction they take an action to   maximize a utility function. In a social setting these decisions are not made in a void, as each agent observes the prediction of her neighbors or their actions, and henceforth updates her prediction. These models are vastly studied by economists who are interested in understanding whether the agents' decisions converge to the same value (consensus), how fast is the rate of convergence, whether an equilibrium exists, and if the consensus leads agents to an optimal decision (learning) ~\cite{AcemglusociallearningBL2011,arieli21social,golub2010naive,golub2017learning, RahimianBL2017,RahimianBLpmlr19a,RahimianSeedingEC19,bindel2015bad}.

Given a social network, various works tackle optimization problems in which an algorithm makes minimal changes in the network to maximize the improvement of a desirable property. 
For instance, various works consider the problem of maximizing information diffusion by seeding information in a number of selected source vertices  
\cite{KempeSeeding2003,adaptiveseedingFOCS2013,RahimianSeedingEC19,garimella2017balancing} or by adding links 
\cite{borgs2014maximizing,d2019recommending}. Some works 
optimally insert links into a network to maximize the information flow  between two groups of nodes \cite{minimizingGionis23,zhu2021minimizing,hittingtimeGionisKDD23,haddadanWSDM21,reducinghaddadan2022,santos2021link}, and  others 
optimaly alter 
innate opinion of users in the FJ model to reduce 
 polarization or  disagreements \cite{tsioutsiouliklis2022link,abebe2018opinion,musco2018minimizing}.

Our work   bridges  the above lines of work. We  study a framework  in which agents are performing learning tasks and exchange their  predictions until each makes a final decision. Unlike prior decentralized learning works, our agents do not necessarily use the same model nor they share any parameters, they solely exchange their predictions.   In this sense, our framework falls into the context of social learning. However, instead of studying problems such as existence of consensus or convergence rates, we focus on the problem of optimally injecting information to a selected subset of agents to improve the overall betterness in the network.    Another difference with social learning framework is that we do not consider one fixed model of information exchange, in contrast, we assume a general linear model for information exchange. Therefore, our methods are applicable to any linear model whether it  describes users of a  social network, each evaluating the truthfulness of an online content, or whether it describes intelligent systems in which agents cooperate for a classification task.
\vspace{-0.2cm}
\subsection{Summary of Contributions}

Motivated by prior works on decentralized and social learning, we consider a new framework called \emph{cooperative learning in a social setting}. In this framework, we consider the problem of \emph{optimally selecting $k$ agents and improving their innate predictions to maximize the overall network  improvement}. We consider both an \emph{aggregate} objective function and an \emph{egalitarian} one, and assume different levels of access to the model's parameters which are  (1) the joint probability distribution of the classifiers forming agents' innate predictions, denoted by $\pi$, and (2) the expressed social influence matrix, denoted by $\lW$.
\vspace{-0.2cm}
\begin{enumerate}
    \item We 
    provide a polynomial time algorithm for the aggregate improvement problem. This algorithm uses only the innate error rates and 
    $\lW$, and it does not need any additional   knowledge of $\pi$; see \cref{sec:algo:agg}. 
    \item We show that solving the egalitarian improvement problem is hard: it is NP-hard to solve it exactly even if both parameters are entirely known. We also show that if $\lW$ and only 
  the innate error rates are available, without any further assumption even finding an approximation solution is hard; see \cref{app:thm:hardness,app:thm:hard:err}~.  
  \item We provide two approximation algorithms for the egalitarian improvement problem. The first algorithm, \algo\, is a greedy algorithm and needs full access to $\pi$. The second one, \appalgo, approximates the greedy choice and only needs access to the innate error rates, but it assume that the innate predictions are pairwise independent. We show that with some modifications \appalgo works under the assumption that the vertices have group dependency; see \cref{sec:algo:egal}.   
  \item We compare the two algorithms for egalitarian improvement by running experiments on real and synthetic graphs and compare their performance to four benchmark methods. 
   Our experiments show that by modifying  only  a few vertices, we succeed in increasing the accuracy of a high percentage of network's agents; see \Cref{fig:eigen-dist-renal}. 
\end{enumerate}


\begin{figure}[h]
    \centering
    \caption{\footnotesize Comparison of \# modified nodes for {Accuracy} $> 90\%$ on different dataset (lower is better).}
    \label{fig:eigen-dist-renal}
\includegraphics[width=0.6\textwidth]{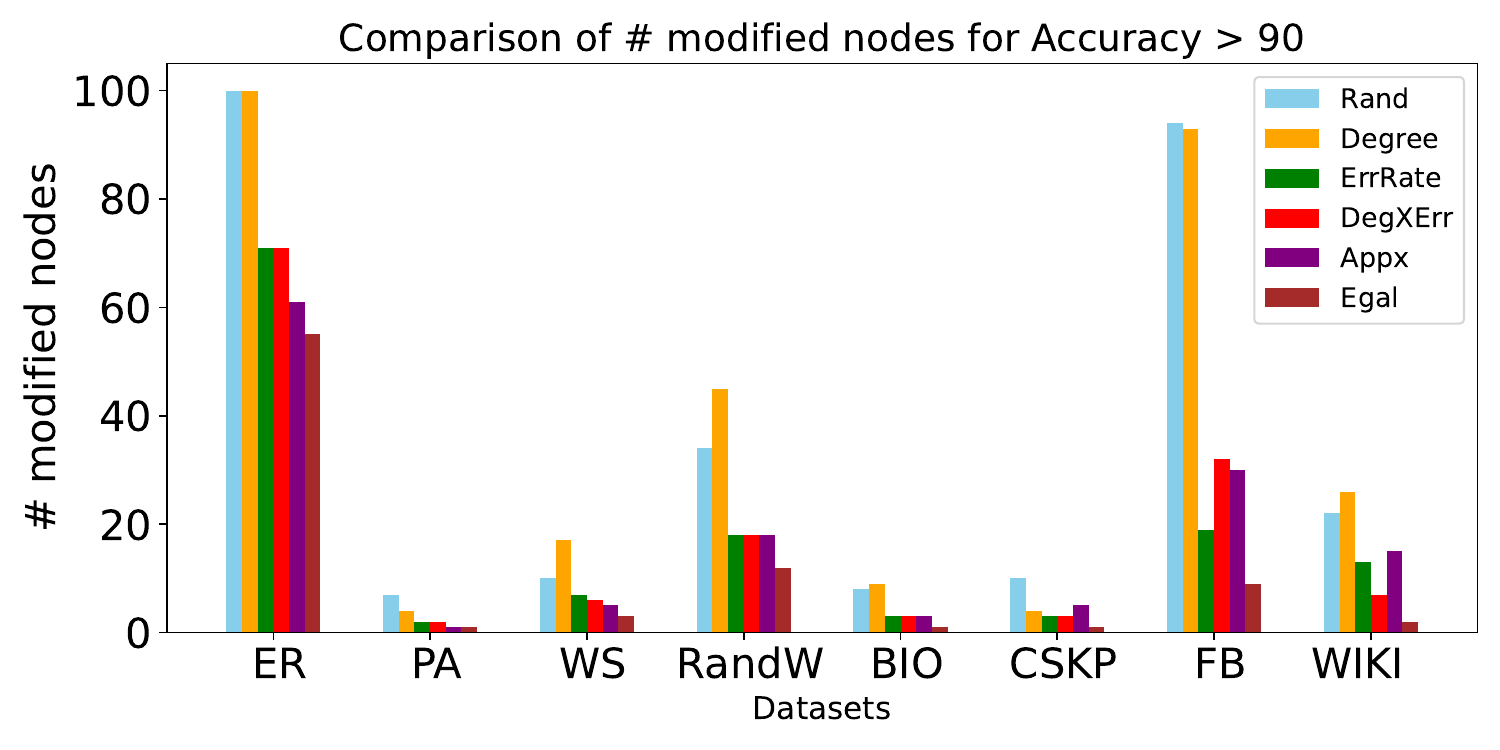}
    \vspace{-2em}
\end{figure}

\section{Models and Problem Definition}
Consider a set $\Omega$ whose elements are labeled as $+1$ or $-1$, i.e., there exists a function $\y: \Omega\rightarrow \{-1,+1\}$ such that for any $a\in \Omega$, $y(a)\in \{-1,+1\}$ is the (true) label of $a$. 
 Consider a set of agents $V=\{v_1,v_2,\dots , v_n\}$ who lack access to the true labels. Given $a\in \Omega$,  each agent $v_i$ uses a classifier $\hy_i$ to predict the label of $a$, i.e., we have $\hy_1,\hy_2,\dots , \hy_n$ where  for any $i$, $\hy_i:\Omega\rightarrow \{-1,+1\}$.
The innate error rate of each agent $v_i$ is defined as:
\begin{equation*}\label{eq:initerr}
 \err(v_i)=\Ex_{a\sim \Omega}\left[ \mathbf{1}\left(\hy_i(a)\neq y(a)\right)\right] =\underset{a\sim \Omega}{\mathbb{P}}\left(\hy_i(a)\neq y(a)\right).
\end{equation*}
We assume that this error is independent of true label.
This assumption  merely makes our analysis simpler and without it our results can be generalized with trivial modifications.

These predictions form the \emph{innate assessment} of the agents, which we  denote by $\z^{(0)}(i,a)$, for arbitrary $i$ and $a$; thus, $\z^{(0)}(i,a)=\hy_i(a) $. After,  agents communicate with each other through a \emph{time varying averaging process}. At each point of time, each agent $v_i$ communicates with a subset of other agents $\N^{(t)}(v_i)$, and her assessment will  update as:
\vspace{-0.15 cm}
\[
\z^{(t)}(i,a)= C_i \cdot \left(\sum_{v_j\in \N^{(t)}(v_i)} W^{(t)}_{i,j} \z^{(t-1)}(j,a)+ \alpha_i  \hy_{i}(a)\right),
\]
where 
$W^{(t)}_{i,j}$ denotes  the influence of $v_j$ on $v_i$ at time $t$, 
$\alpha_i$ is 
$v_i$'s stubbornness, $C_i$ is a normalizing constant and  $\z^{(t-1)}(i,a)$  and $\z^{(t)}(i,a)$ respectively denote $v_i$'s assessment before the $t$th round of communication and after it.

Formally, given $a\in \Omega$,
let $\mathbf{\hy}(a)=\left(\hy_1(a),\hy_2(a),\dots \hy_n(a)\right)$ and $\vz^{(0)}(a)=\mathbf{\hy}(a)$.
 Assume having a sequence of $n\times n$ 
  matrices $\left(W^{(t)}\right)_{t=1}^\infty$
  , $\boldsymbol{\alpha}=(\alpha_1,\alpha_2,\dots ,\alpha_n)$ and $\boldsymbol{C^{(t)}}=(C^{(t)}_1,C^{(t)}_2,\dots ,C^{(t)}_n)$. For any $t\geq 1$ we have: 
\begin{equation}\label{eq:generalmodel}
    \vz^{(t)}(a)=
    \boldsymbol{C^{(t)}}\odot 
    \left(W^{(t)}\vz^{(t-1)}(a)+ \boldsymbol{\alpha}\odot  \mathbf{\hy} (a)\right) ~,
\end{equation}
where $\odot$ denotes the Hadamard product which is defined to be the vector (matrix) obtained from the pairwise  product of the elements in two other vectors (matrices).
\footnote{
Let $A,B\in \mathbb{R}^{n\times m}$, and $C=A\odot B$. We have that $C\in \mathbb{R}^{n\times m} $ and $C_{ij}=A_{ij}\cdot B_{ij}$
}
Let $\vz^{*}(a)=\left(
z^*(1,a),z^*(2,a),\dots, z^*(n,a)\right)$ be the vector that the above process converges to i.e., 
\[
\vz^{*}(a)= \lim_{t\rightarrow \infty }\vz^{(t)}(a)~.\] 
We call $\vz^{*}(a)$ the \emph{the expressed prediction} vector. 
 We assume that this limit has the following closed form: 
\begin{equation}\label{eq:convergence}
\exists \lW\in \mathbb{R^+}^{n\times n} \ \text{ s. t. } \forall a\in \Omega,\  
 \vz^{*}(a)=\lW\mathbf{\hy}(a)^{\mathsf T} ~.
\end{equation}
 $\mathbb{R}^+$ being the set of all non-negative real numbers. We call $\lW$ the \emph{expressed (social) influence} matrix.

The above assumption is indeed proven to hold true in many well studied models. To name a few:

\vspace{-0.25cm}
\paragraph{DeGroot Model~\citep{DeGroot1974-ed}} Given a row-stochastic  matrix ${W\in \mathbb{R}^+}^{n\times n}$, DeGroot Model evolves as: 
\begin{equation*}\label{eq:Degroot}
\vz^{(t)}_{\textsc{DeGroot}}(a)=W\vz^{(t-1)}_{\textsc{DeGroot}}(a), ~ \vz^{(0)}_{\textsc{DeGroot}}(a)=  \mathbf{\hy}(a),
\end{equation*}
It is known that 
\vspace{-0.3 cm}
\begin{equation*} \vz^{*}_{\textsc{DeGroot}}(a)= \Pi \mathbf{\hy}(a)~,
\end{equation*}
where $ \Pi $ is a block diagonal matrix with blocks $\Pi_1, \Pi_2 , \dots,\Pi_k$, and each block corresponds to a strongly connected component of the graph corresponding to $W$. In each block all rows are identical. 

\vspace{-0.2cm}
\paragraph{Friedkin–Johnsen (FJ)  Model~\citep{Friedkin1990-wl}}
Given an arbitrary  matrix with non-negative weights $W$, 
for each $i\in [n]$ let $C_i=1/(1+\sum_{j\in \N(i)}W_{i,j})$.
FJ model evolves as follows:
\begin{align*}\label{eq:FJmodel}
  \vz^{(t)}_{\textsc{FJ}}(a)=\boldsymbol{C} \odot \left(W\vz^{(t-1)}_{\textsc{FJ}}(a)+    \mathbf{\hy}(a)\right), \ \ 
  \vz^{(0)}_{\textsc{FJ}}(a)=\mathbf{\hy}(a). 
\end{align*}
It is known that if $W$ is symmetric the FJ model converges to the following closed form:
\vspace{-0.2cm}
\[\vz^{*}_{\textsc{FJ}}(a)=(I+L)^{-1}\mathbf{\hy}(a), 
\]
where $L$ is the combinatorial Laplacian of the indirected graph corresponding to $W$ and $I$ denotes the identity matrix. 

Both DeGroot and FJ model use a constant matrix $W$ in their evolution. The following shows a time varying evolution:
\vspace{-0.2cm}
\paragraph{Finite time Models}
Assume a finite sequence of row stochastic matrices 
$W^{(1)}, W^{(2)},\dots , W^{(T)}$, for any $t>T$, let $W^{(t)}$ be the identity matrix and let $\boldsymbol{\alpha}$ be all zeros vector. 
It is straightforward to see that the general equation of \cref{eq:generalmodel} will converge to 
\vspace{-0.4cm}
\begin{equation*}\label{eq:finite}    \vz^{*}_{\textsc{Finite}}(a)=\lW {\mathbf \hy}(a), ~
\lW=\prod_{i=1}^T W^{(i)}~.
\end{equation*}
\vspace{-0.4cm}
\subsection{Statement of Problems} 

Assume any time varying averaging model which converges to a close form as \cref{eq:convergence}. Thus, for a given $a\in \Omega$ and   $v_i\in V$  the \emph{expressed} prediction of $v_i$ on $a$ is equal to 
\vspace{-0.25cm}
\[
\z^*(i,a)=\sum_{j=1}^n \lW_{ij}\hy_j(a)~.
\]
Note that 
the value of $z^*(i,a)$ is a function of the expressed social influence matrix $\lW$ as well as the quality of all agents' classifiers which is formulated in the joint probability distribution of  $\mathbf{\hy}$. 
If the true label $\y(a)$ equals $+1$, the positive values of $\z^*(i,a)$ are preferable, otherwise negative values of $\z^*(i,a)$ are better. In other words, we would like $\y(a)$ and $\z^*(i,a)$ to have the same sign. We define
$\B(i,a):V\times \Omega\rightarrow [-1,1]$ 
as follows: 
\[
\B(i,a)=\y(a)\cdot \z^*(i,a)~.
\]
The larger values for $\B(i,a)$ correspond to the fact that agent $v_i$ makes a good prediction on an object $a$. If $\B(i,a)<0$ we consider this prediction \emph{faulty}. 
\paragraph{Improving quality of selected classifiers} 
Assume that we have a tool to  improve the quality of  classifiers. Formally, let  $\imp\in(0,1]$ be a given constant.
For any arbitrary agent $v_i$ we may improve  $\hy_i$ to $\ty_i:\Omega\rightarrow [0,1]$ as follows: 
\vspace{-0.1cm}
\begin{equation}\label{eq:improvementphi}
 \forall a\in\Omega,~
\ty_i(a)= (1-\imp)\hy_i(a)+ \imp \y(a)~.   
\end{equation}
We would like to improve the quality of a selected subset of agents' classifiers (innate predictions) to maximize the overall quality of expressed predictions among all agents. Formally 
 we are interested in selecting a subset $S$ of   $k$ agents, i.e.,  $S\subseteq V$, $\abs{S}=k$
and improve the quality of the classifier as follows: 
\begin{equation}\label{eq:intervention}
  \forall a\in\Omega,~
\ty_i(a)= 
\begin{cases}
(1-\imp)\hy_i(a)+ \imp y(a)~& \text{ if } v_i\in S\\
\hy_i(a)& \text{ if } v_i\notin S\\
\end{cases}   
\end{equation}
Let 
\vspace{-0.4cm}
\begin{equation*}
  \znew^*(i,a)=\sum_{j=1}^n \lW_{ij}\ty_j(a)\ \ \& ~\Bnew(i,a)=\y(a)\cdot \znew^*(i,a)~.  
\end{equation*}
In the first problem that we study, $S$ is selected to maximize an \emph{aggregate} objective function:
\vspace{-0.2cm}
\begin{equation}\label{eq:agg}
    \Gagg(S)\triangleq  
\Ex_{a\sim\Omega}\left[ \sum_{i=1}^n \Bnew(i,a)-\B(i,a)
\right]
\end{equation}
The above objective function is great, but it has the shortcoming of any other aggregate objective function: it is possible that one agent  benefits  enormously from it at the cost of many other agents getting extremely little. 

Therefore, we also study  the following \emph{egalitarian}  objective function in which we  count the expected  \emph{number} of agents whose faulty predictions will improve.   
\vspace{-0.15cm}
\begin{equation*}\label{eq:egal}
\Gegal(S)
{\triangleq} 
\Ex_{a\sim\Omega} \left[\sum_{i=1}^n \mathbf{1}(\B(i{,}a){<}0 \wedge \B(i{,}a){<}\Bnew(i{,}a))\right].
\end{equation*}
We are now ready to 
formally state the problems. 

\smallskip

\begin{problem}[Aggregate improvement through improving $k$ selected agents]\label{problem1}
 What is an optimal way to select a subset $S\subseteq V$ and update their innate predictions as \cref{eq:intervention} to maximize the following objective function: 
 \vspace{-0.15cm}
\[
\OPT{agg}=\underset{S\subseteq V\\ ; |S|=k} {\rm{max} }\ \Gagg(S)~.  
\]
\end{problem}

\begin{problem}[Egalitarian improvement through improving $k$ selected agents]\label{problem2}
 What is an optimal way to select a subset $S\subseteq V$ and update their innate predictions as \cref{eq:intervention} to maximize the following objective function:
  \vspace{-0.15cm}
\[
\OPT{egal}=\underset{S\subseteq V\\ ; |S|=k} {\max }\ \Gegal(S)~.  
\]
\end{problem}
  \vspace{-0.3cm}
The above process has two main parameters: a joint probability distribution  $\pi: \Omega\times \{-1,+1\}^V\rightarrow [0,1]$, where 
for any $a\in \Omega$ and $\vec{b}\in \{-1,+1\}^n$, 
$\pi(a,\vec{b})=\Prob(\wedge_{i=1}^n \hy_i(a)=b_i)$ as well as $\lW$ which is the expressed influence matrix. 
While in most applications $\lW$ is either available in closed form (e.g., for DeGroot, FJ or finite models) or it can be approximated using iterative methods, our access to $\pi$ depends on  assumptions which may vary depending on our application. In fact, we present our algorithms assuming different access levels to the joint probability distribution $\pi$. 

\medskip 

\section{Summary of Results}\label{sec:summary}




In this section, we present a summary of our main results. Let us first list the assumptions we make on the access level to $\pi$ and $\lW$:

\begin{assumption}[Best scenario]\label{assu:bestcase}
Assume having \emph{complete} knowledge of $\pi$.
\end{assumption}
The above assumption is reasonable if $\Omega$ is a small finite set. For instance in  cases where $\Omega$ can be partitioned to a few types and the agents make the same predictions on any element of the same type, e.g., see  \cite{demarzo2003persuasion,Hazla2019-df,Gaitonde2021-vk} 

Clearly, it is possible that the above assumption does not hold true. However, the algorithm needs to have \emph{some knowledge} of the  probability distribution of innate predictions. 

\smallskip 

\begin{assumption}[some knowledge of $\pi$]\label{assu:relaxpi1}
Assume having some knowledge of  classifiers' dependence/independence and the innate error rates.
\end{assumption}

\begin{assumption}[minimum knowledge of $\pi$]\label{assu:relaxpi2}
Assume having only knowledge of  innate error rates $\{\err(v_i)\}_{v_i\in V}$.
\end{assumption}

The summary of our main results is that \cref{problem1} is easy, i.e., we show an exact polynomial time  algorithm for it assuming   minimum knowledge of $\pi$. 
On the other hand, we show that \cref{problem2} is hard, i.e., we show that exactly solving it is NP-complete even assuming full knowledge of $\pi$. Note that this hardness also holds for the cases in which we have less knowledge of $\pi$. 
Later, we show greedy algorithms for approximately solving it under different assumptions. 

Initially, in all of our results we assume access to the closed form of $\lW$.
We then 
show that the guarantees still hold with a slight change when an approximation of $\lW$ is given.

\begin{assumption}[knowledge to an approximation of  $\lW$]\label{assu:relaxW}
 Assume   having knowledge of $\hW$ 
such that 
\[
\left| \hW-\lW \right|_1\leq \epsilon ~,
\]
where  $\left|\cdot \right|_1$ is the $\ell_1$ norm and $\epsilon$  a precision parameter.  
\end{assumption}

\subsection{Optimizing the aggregate objective function}

\begin{theorem}\label{thm:agg1}
    There is an algorithm with run-time complexity $\Theta\left(n^2\right)$  which given  $\lW$ and $\{\err(v_i)\}_{i=1}^n$ as input parameters outputs $S$ such that $\Gagg(S)=\OPT{agg}$.
\end{theorem}

\begin{remark}\label{thm:agg2}
      Let $\mathsf{ALG}$  be the algorithm whose performance guarantees are presented in  \cref{thm:agg1}. Let $S$ be the output of $\mathsf{ALG}$ when $\lW$ is given to it as input parameter, and let $S'$ be the output if $\widehat{\lW}$ is given. We have:
    \[ \Gegal(S')\geq \Gegal(S)-2k\epsilon~.\]
\end{remark}
We present this 
 algorithm in \cref{sec:algo:agg,app:algo:agg}.

\vspace{-0.2cm}
\subsection{Optimizing the egalitarian objective function}

We now present our results about \Cref{problem2}. 

We call a matrix with no negative  entry  \emph{
non-negative}. 
In \textsc{DeGroot} model if $W$ is 
non-negative, $\lW$ is also 
non-negative and in \textsc{FJ} model if  $W$  is 
non-negative and symmetric, $\lW$ is also 
non-negative \citep{chebotarev2006proximity}
.

In this section we assume that $\lW$ is non-negative.
\smallskip 

\begin{theorem}\label{thm:hardness}
Under \cref{assu:bestcase} and assuming $\abs{\Omega}$ is polynomial in $n$ ,
 \Cref{problem2}  is NP-hard.
\end{theorem}
Since \Cref{problem2} is NP-hard, we concentrate on finding an approximation algorithm for it in different scenarios. 
\vspace{-0.2cm}
\paragraph{Approximate solution with full access to $\mathbf \pi$.}
Assume that $\Omega$ is finite and for any $a\in \Omega$ and $b\in \{-1,+1\}^n$ we can evaluate the probability of the event $\bigwedge_{i=1}^n \hy_i(a)=b(i)$.

\begin{theorem}\label{thm:app:egal}
There is a greedy algorithm with runtime $\Theta\left(\vert\Omega\vert n^2k\right)$ which by receiving  $\pi$ and $\lW$ as input parameters 
outputs $S$ satisfying 
\[
\Gegal(S)\geq (1-1/e)\OPT{egal}.
\]
\end{theorem}

A pseudocode of our algorithm, \algo\ is presented in \cref{app:sec:psudocodes}, an overview of main ideas is presented \cref{sec:algo:egal}. Note that the runtime of \algo\ grows linearly in $\abs{\Omega}$. Later, we present \appalgo whose  complexity  \emph{does not} grow with  $\abs{\Omega}$ (see \cref{sec:algo:egal}). Therefore, even if $\pi$ is fully known, by employing \appalgo, we may prefer to suffice to a lower quality approximation to gain better time complexity. 

\paragraph{Approximate solution with minimum information} 
While the previous result shows a constant approximation,  the following theorem shows that  by only having the innate error rates $\{\err(v_j)\}_{j=1}^n$ and $\lW$ we are not able to achieve any good approximation.  
\begin{theorem}
Any  solution to 
\cref{problem2} which only uses $\lW$ and innate error rates $\{\err(v_j)\}_{j=1}^n$ makes an error which can be as large as $\Omega(n)$.
\end{theorem}
We restate and prove the above theorem in \cref{app:thm:hard:err}. 

Motivated by the above results, we now present our results when more knowledge about the classifiers are available.  For instance, in addition to knowing the error rates we  assume that the classifiers are pairwise independent.

Using these assumptions we  design \appalgo\ and show theoretical guarantees for its performance. 
\vspace{-0.2cm}
\paragraph{Approximate solution assuming independence. }
Assume that the classifiers $\{\hy_i\}_{v_i\in V}$ are pairwise independent. We have: 

\begin{theorem}\label{thm:indep}
There is a greedy algorithm, \appalgo, with run-time $\Theta(n^3k)$ which by receiving  $\{\err(v_i)\}_{v_i\in V}$ and $\lW$ as input parameters 
outputs $S$ satisfying 
\[
\Gegal(S)\geq [(1-1/e)-\Delta^{\rm ind}]\cdot \OPT{egal}~,
\]
where $\Delta^{\rm ind}$ is a parameter depending on the network. If the network is \emph{nicely structured}  $\Delta^{\rm ind}=o(1)$;  see \cref{sec:algo:ind}~.
\end{theorem}
We generalize the assumption of pairwise independence to group dependence as follows:
\vspace{-0.1cm}
\paragraph{Approximate solution assuming  group dependence. }

Assume  that some agents belong to opposing groups $\R$ and $\Bl$ and some agents are colorless; they are in $\W$. The classifiers of the agents in  $\W$ are independent, and the classifiers of $\R$ and $\Bl$ agents have positive intra-correlation and negative inter-correlation as described in \cref{def:model}. This model describes a situation when have a classification task that can be influenced by an exogenous factor, e.g, their position or  political leaning. Clearly by setting $V=\W$ we will have the previous model.  In this case we have:

\begin{theorem}\label{thm:group}There is a greedy algorithm, \appalgo, with run-time $\Theta(n^3k)$ which by receiving  \emph{individual and group} error rates and $\lW$ as input parameters 
outputs $S$ satisfying 
\[
\Gegal(S)\geq [(1-1/e)-\Delta^{\rm gr}]\cdot \OPT{egal}~,
\]
where $\Delta^{\rm gr}\geq \Delta^{\rm ind}$ is a parameter depending on the network and the dominance of colors $\R$ and $\Bl$ on other agents. Not surprisingly,  $\Delta^{\rm gr}$ becomes closer to $\Delta^{\rm ind}$ as the number of colorless agents increases. If the network is  \emph{nicely structured} and \emph{nicely colored}  $\Delta^{\rm gr}=o(1)$; See \cref{sec:algo:group}~.
\end{theorem}

The following remark holds  ture both under pairwise Independence and group dependency: 
\begin{remark}\label{remark:appW}
    Let $S$ be the output of \appalgo when $\lW$ is given to it as input parameter, and let $S'$ be the output if $\widehat{\lW}$ is given. We have:
    \vspace{-0.2cm}
    \[ \Gegal(S')\geq \Gegal(S)(1-8 k \epsilon)~.\]
\end{remark}

\section{Algorithms}\label{sec:algos}

In this section we present our algorithms. 
All of the algorithms are greedy. We provide an exact solution for \cref{problem1} in \cref{sec:algo:agg}. In \cref{sec:algo:egal}, because of the NP-harness of  \cref{problem2}, we present  
a  \emph{constant} approximation algorithm for it; we call this algorithm \algo\ . We then present 
\appalgo\  which has a  lower time complexity but worst approximation guarantees assuming pairwise independence  of agents. In this case, our approximation ratio depends on the network structure. In \cref{sec:algo:group}, we modify \appalgo\  so that it works under a milder assumption formalized as \emph{group dependency}.

Pseudocodes for our egalitarian algorithms are presented in \cref{app:sec:psudocodes}

\subsection{The aggregate objective function}\label{sec:algo:agg}
For any vertex $v_j$ let's define the influence score and its approximation by 
${\rm Inf}(v_j)=\sum_{i=1}^n \lW_{ij} \err(v_i) ~.$
The following lemma is proven in \cref{app:algo:agg}.

\begin{lemma}\label{lem:agg1}
Let $U=\{u_1,u_2,\dots, u_k\}$ be  top-$k$ vertices with highest value of $\sum_{u_j\in U}{\rm Inf}(u_j)$. We have that:
\vspace{-0.2cm}
\[
\Gagg(U)=\OPT{agg}~.
\]
\end{lemma} 
\paragraph{Proof of \Cref{thm:agg1} and \Cref{thm:agg2}}
With the above lemma, we design an algorithm  that for all $v_j$s calculates  their influence score,  and outputs the top-$k$. 
 The  complexity of such algorithm   is $\Theta(n^2+n\log n +k)=\Theta(n^2)$. In \cref{app:algo:agg} we also prove 
 \Cref{thm:agg2}.

\subsection{The egalitarian objective function}\label{sec:algo:egal}

Optimizing the egalitarian function is NP-hard (See \cref{app:thm:hardness}). We show that $\Gegal:2^V\rightarrow [0,n]$ is monotonic and sub-modular (See \cref{app:lem:monotone,app:lem:submodular}). Thus, 
a  greedy algorithm will provides a $(1-1/e)$ approximation \cite{nemhauser1978best}.  

We now concentrate on obtaining the greedy choice. 
Formally, we define the function $\Greedy(S):2^V\rightarrow V$ as follows:
\vspace{-0.3cm}
\begin{equation}\label{def:greedy}
  \Greedy(S)=\underset{u\in V}{\rm argmax } \ \Gegal(S\cup\{u\})- \Gegal(S)~.   
\end{equation}
The following lemma is proven in \cref{app:greedyfind}:

\begin{lemma}\label{lem:greedy}
For any $S\subseteq V$ we have:
\begin{equation}\label{eq:sumgreedy}
\Greedy(S)=\underset{u\in V}{\rm argmax }
\sum_{\substack{i=1:n
\\
\lW_{iu}\neq 0
}}  \Gain_i(S,u)~,
\end{equation}
where $\Gain_i(S,u)$ is defined to be\footnote{We use  $\wedge$ and $\vee$  to denote respectively the logical operations conjunction and disjunction which are  \emph{``and''} and  \emph{``or''}.} 
\begin{equation*}\label{eq:gainfull}
\Prob_{a\sim \Omega}
\left(
\B(i,a)\leq 0  \wedge \left(\bigwedge_{
\substack{v_j\in S
\\ 
\lW_{ji}\neq  0 
} } \y(a){=} \hy_j(a)\right)
\wedge \y(a){\neq} \hy_u(a)
\right)
\end{equation*}
\end{lemma}

\paragraph{Proof of \cref{thm:app:egal}}
Our proposed algorithm, 
\algo\ starts with $S=\emptyset$. Iteratively,  $\Greedy(S)$ is added to $S$ until $\abs{S}=k$.  
Assume that $\Omega$ is finite and we have access to $\pi$. Using \cref{lem:greedy} we obtain the greedy choice as follows:   for any $a\in \Omega$ and $v_i\in V$, we evaluate the validity of the event $\left(\bigwedge_{
\substack{v_j\in S
\\ 
\lW_{ji}\neq  0 
} } \y(a){=} \hy_j(a)\right)
\wedge \y(a){\neq} \hy_u(a)$. If this event is valid,  we calculate $\B(i,a)$ and verify $\B(i,a)<0$ which takes $n$ steps.  
We find best $u$  by using \cref{eq:sumgreedy} and iterating over all choices of $i\in V$ and $a\in \Omega$. 
The total runtime for $k$ iterations  is $\Theta(\abs{\Omega}n^2k)$.

\subsubsection{Independent classifiers}\label{sec:algo:ind}

In this section we consider the case where the classifiers $\{\hy_i\}_{v_i\in V}$ are pairwise independent which falls into the scenario 
in which we have \emph{some} knowledge of $\pi$ (Assumption~\ref{assu:relaxpi1}). 
In this case we may estimate the greedy choice as a function of $\{\err(v_i)\}_{i=1}^n$ and $\lW$. 

Let's state the main result of this section and then we present the steps that lead us to the selection of the greedy choice:

\begin{theorem}\label{thm:egalcorerct}
Let $S$ be the output of a greedy algorithm which starts by taking $S=\emptyset$ and for $k$ steps keeps updating $S$ to  $S\cup\{g\}$ where 
\vspace{-0.2cm}
\[g=\underset{u}{\rm argmax}\sum_{\substack{i=1:n
\\
\lW_{iu}\neq 0
}} \widehat{\Gain}_i(S,u)~, \]
and $\widehat{\Gain}_i(S,u)$ is  defined in $\cref{lem:appgain}$, we have: 
\[
\Gegal(S)\geq [(1-1/e)-\Delta^{\rm ind}]\cdot \OPT{egal}~,
\]
with $\Delta^{\rm ind}=\Theta(\abs{{\mathcal A}})$ and $\mathcal A$ is the set of ambiguous vertices. 
\end{theorem}
\paragraph{Ambiguous vertices} Consider the partitioning of $V$ with $V^+$ as low error vertices and $V^-$ as high error vertices:
\[V^+=\{v_j\mid \err(v_j)\leq 1/2\} \; \& \ V^-=\{v_j\mid \err(v_j)> 1/2\}\]
with respect to this partition we define the following vectors whose elements are in $[0,1]$:
\[
 {\mathcal E}^+= \left(1-2\err(v_j)\right)_{v_j\in V^+}\ \ \& \  \  {\mathcal E}^-= \left(2\err(v_j)-1\right)_{v_j\in V^-}
\]
For any arbitrary vertex $v_i\in V$, low error and high error vertices contribute in the value of  $\B(i,a)$ through the following coefficient vectors: 
\[
\lW^+_i=(\lW_{ij})_{v_j\in V^+} \quad   \&  \  \ \lW^-_i=(\lW_{ij})_{v_j\in V^-}
\]
The ambiguous vertices are those who \emph{are not} dominated by neither $V^+$ or $V^-$.
 Formally, 

\begin{definition}\label{def:nice}
    [Ambiguous vertices]
Let $\lW_i=(\lW_{i1},\lW_{i2},\dots, \lW_{in} )$,  $\abs{\cdot}_2$ denote the  $\ell_2$ norm and $\langle \cdot, \cdot \rangle$  dot product.  
    We call a vertex $v_i\in V$ \emph{ambiguous} if it satisfies: 
    \[\abs{
    \frac{
 \langle \lW_i^+, {\mathcal E}^+\rangle 
 }{\abs{\lW_{i}}_2} 
 - \frac{
 \langle \lW_i^-, {\mathcal E}^-\rangle 
 }{\abs{\lW_{i}}_2}} \leq 4\sqrt{ \log n}~.
    \]
    A network is \emph{nicely structured} if it has no ambiguous vertex. 
\end{definition}

If a vertex is \emph{non-ambiguous} we can estimate  the gain associated to it very precisely: 

\begin{lemma}\label{lemm:egalapp:good}
Let $\Gain_i(S,u)$ and $\widehat{\Gain_i}(S,u)$  be as defined respectively as in  \cref{eq:gainfull}
and \cref{lem:appgain}.
If a vertex is \emph{non-ambiguous} we have: 
\[
\abs{\Gain_i(S,u)-\widehat{\Gain_i}(S,u)}\leq o(n^{-1})
~.\]
\end{lemma}
We now present the following lemma  related to the approximation of greedy choice. All the proofs and details are presented in \cref{app:sec:ind}.


\begin{lemma}\label{lem:appgain}
Let $\Gain_i(S,u)$ be as \cref{eq:gainfull}.
    Let $  \widehat{\Gain_i}(S,u):2^V\times V\rightarrow [0,1]$ be defined as follows: 
    \[
     \widehat{\Gain_i}(S,u)\triangleq
        \mathbf{1}\left(\val_i(S,u)<0\right)\err(u)\prod_{\substack{{v_j\in S}\\
           \lW_{ij}\neq 0}}(1-\err(v_j))~.
    \]
\vspace{-0.2cm}
    We have:
    \[
        \abs{\widehat{\Gain_i}(S,u)-\Gain_i(S,u)}\leq  \exp\left(
    -\frac{\val_i(S, u)^2}{4\sum_{i=1}^n\lW_{ij}^2}
    \right)~,
       \]
    \   
    where 
    \[
    \val_i(S,u)=-\lW_{iu}+\sum_{v_j\in S}\lW_{ij}+\sum_{\substack{
    j=1:n\\
    j\not \in S\cup\{u\} 
    }}\lW_{ij}[1-2\err(j)]~.
    \]

\end{lemma}

\paragraph{Proof of  \cref{thm:indep,remark:appW}} 
A complete pseudocode of our proposed algorithm,  \appalgo, is presented in \cref{app:sec:psudocodes} (\cref{alg:2c}).  It is easy to see that the runtime is dominated by $\Theta(n^3k)$. Note that the correctness of \cref{thm:indep} is directly concluded from \cref{thm:egalcorerct} by setting $\abs{{\mathcal A}}=0$. We present the proof of \cref{thm:egalcorerct,remark:appW} in \cref{app:sec:proofs:egal}. 

\subsubsection{Group dependent classifiers}\label{sec:algo:group}

We now consider a case when agents are either \emph{red}, \emph{blue} or none (\emph{white}). The agents who are red or blue agent either  all follow a group decision
, or they independently follow an individual decision. The group decision of the blue agents is always the opposite of the group decision of red agents. 


\begin{definition}\label{def:model}
[Group Dependence]
Assume that the set of agents $V$ can be partitioned as $V=\R\cup \Bl \cup \W$. We assume a set of classifiers  $\ihy_1, \ihy_2,\dots , \ihy_n: \Omega\rightarrow \{-1,+1\}$ which are pairwise independent. Furthermore, we assume  two group classifiers $\ghy_{\R},\ghy_{\Bl}: \Omega\rightarrow \{-1,+1\}$ which satisfy: 
\vspace{-0.2cm}
\[
\forall a\in \Omega, \ \ghy_{\R}(a)\neq \ghy_{\Bl}(a)~.
\]
Given a constant $\rho\in [0,1]$, these classifiers construct  $\{\hy_1,\hy_2,\dots, \hy_n\}$ as follows: 

With probability $\rho$ the red and blue agents follow their groups' decision, i.e., 
\vspace{-0.2cm}
\[
\forall v_i\in \R,\ \hy_i(a)= \ghy_{\R}(a) \wedge 
\forall v_i\in \Bl,\ \hy_i(a)= \ghy_{\Bl}(a)
\]
And  the white agents independently follow their individual decisions, i.e, for all $v_i\in \W$,\ 
$
\hy_i(a)= \ihy_{i}(a) ~.
$

Alternatively, with probability $1-\rho$ , all agents independently use their individual classifiers. i.e, 
\[
\forall v_i \in V, \ \hy_i(a)= \ihy_{i}(a)~. 
\]
\end{definition}
In the above setting we use the following notation:
For any $v_j\in V$ we define 
$\err^{\rm indv}(v_j)=\Prob\left(\ihy_j(a)\neq \y(a) \right)$, and 
\[\err(\R)=\Prob\left(\ghy_\R(a)\neq \y(a) \right)\ \& \
\err(\Bl)=\Prob\left(\ghy_\Bl(a)\neq \y(a) \right)
\]
It is immediate from the definition that 
$1-\err(\Bl)=\err(\R)$.

In this setting, the estimation of $\Gain_i(S,u)$ is more involved 
and is presented in \cref{sec:algo:group}. In this case, our greedy algorithm uses $\{\err^{\rm indv}(v_j)\}_{j=1}^n$, $\err(\R)$, $\err(\Bl)$ and $\lW$ or its approximation. The final result follows:

    \begin{theorem}\label{thm:egalgroupcorrectness}
Let $S$ be the output of a greedy algorithm which approximates greedy choice as defined in $\cref{app:greedy:group}$. We have: 
\[
\Gegal(S)\geq [(1-1/e)-\Delta^{\rm gr}]\cdot \OPT{egal}~,
\]
where 
$
\Delta^{\rm gr}=\Theta(\rho \abs{{\mathcal A}^{\W}}+ (1-\rho)\abs{{\mathcal A}}
)$, $\mathcal A$ is the set of ambiguous vertices defined before and ${\mathcal A}^{\W}$ is the set of \emph{$\W$-ambiguous}  vertices. 

\end{theorem}
\vspace{-0.3cm}
\paragraph{$\boldsymbol{\W}$-Ambiguous vertices.} As in \cref{sec:algo:group}, we partition $\W$ to low error vertices $\W^+$ and high error vertices $\W^-$. Similarly we define $ {\mathcal E^\W}^+$, $ {\mathcal E^\W}^-$ and for any $v_i\in V$, ${\lW^{\W+}_i}$ and ${\lW^{\W-}_i}$; for details see \cref{app:sec:Wambig}. We define: 
\begin{definition}\label{def:nicecolor}
    [ $\W$-Ambiguous vertices ]
Let $\lW^\W_i=(\lW_{ij})_{v_j\in \W}$, and $\abs{\cdot}_2$ be the $\ell_2$ norm and $\langle \cdot, \cdot \rangle$ be dot product.

    We call an agent $v_i\in V$, \emph{$\W$-ambiguous}  if it satisfies 
 \[\abs{
 \frac{
 \langle{\lW^{\W+}_i},  {\mathcal E^\W}^+ \rangle 
 }{\abs{\lW_i^\W}_2}-
  \frac{
 \langle{\lW^{\W-}_i},  {\mathcal E^\W}^- \rangle 
 }{\abs{\lW_i^\W}_2}
 }\leq 4\sqrt{\log n}+ \Delta\lW_i~,
    \]
    where
$\Delta\lW_i=\abs{\sum_{v_j\in \R}\lW_{ij} - \sum_{v_j\in \Bl}\lW_{ij}}~.
    $
    A \emph{network} is \emph{nicely colored} if no vertex is $\W$-ambiguous.
\end{definition}
\vspace{-0.25cm}
\paragraph{Proof of \cref{thm:group,remark:appW}} 
Pseudocode and details 
are presented in  \cref{app:sec:algo:group,app:sec:psudocodes,app:thm:proofs:group}.
\Cref{thm:group} can be directly concluded from \cref{thm:egalgroupcorrectness} by setting $\abs{\mathcal A}=\abs{{\mathcal A}^\W}=0$.


\section{Experiments}
In this section, we empirically validate the effectiveness of our proposed methods for \cref{problem2} through a series of meticulously designed experiments.
We test our algorithms {\algo} and {\appalgo} (pseudocodes in \cref{app:sec:psudocodes}) 
 benchmarked against random selections together with three heuristic methods: selecting nodes based on degree (Degree), innate error rate (ErrRate), and the product of degree and error rate (DegXErr).
 
 

\vspace{-0.35cm}
\paragraph{Synthetic Datasets}
Synthetic data is generated with three components: a random graph $G$, a weight matrix $\bar{W}$ and initial opinions $\mathbf{\hy}$. To generate  $G$, 
we employ 
Erd\H{o}s-R\'{e}nyi model (ER)~\cite{erdos59a}, Barab\'{a}si-Albert model (PA) \cite{barabasi1999emergence} and Watts-Strogatz model (WS) \cite{watts1998collective}.
We generate the weight matrix $\bar{W}$ using the FJ model. 
Each $\hy_i(a)$ is sampled from Bernoulli distribution with a randomly chosen $p_i$. 

\begin{table*}[h!]
  \centering
\caption{Comparison of experiments on five methods Egal={\algo}, Appx={\appalgo}, Rand={Random selection}. 
}
\label{tb:exp}
\begin{footnotesize}  
  \begin{tabular}{|p{1.3cm}|p{1.2cm}|p{1cm}|p{1cm}|p{1cm}|p{1cm}|p{1cm}|p{1cm}|p{1cm}|p{1cm}|}
  \hline
                        &               & \multicolumn{8}{c|}{\textbf{Datasets}}    \\ \hline
  \textbf{Score}        & \textbf{Method} & {ER (128)} & {PA (128)} & {WS (128)} & {RandW (128)} & {BIO (297)} & {CSPK (39)} & {FB (620)} & {WIKI (890)} \\ \hline
  \multirow{6}{*}{\makecell{{\cvratio @ } \\ k=log(n)}} & Rand          & 0.11         & 0.88         & 0.53         & 0.18              & 0.80          & 0.63          & 0.35         & 0.48           \\ \cline{2-10}
                        & Degree        & 0.08         & 0.96         & 0.42         & 0.12              & 0.78          & 0.84          & 0.36         & 0.49           \\ \cline{2-10}
                        & ErrRate       & 0.22         & \textbf{1.00} & 0.76         & 0.47              & 0.96          & 0.94          & 0.53         & 0.54           \\ \cline{2-10}
                        & DegXErr       & 0.18         & \textbf{1.00} & 0.89         & 0.37              & 0.96          & \textbf{1.00} & 0.72         & 0.78           \\ \cline{2-10}
                        & Appx          & 0.18         & \textbf{1.00} & 0.87         & 0.41              & 0.94          & 0.84          & 0.62         & 0.64           \\ \cline{2-10}
                        & Egal          & \textbf{0.27} & \textbf{1.00} & \textbf{1.00} & \textbf{0.58}     & \textbf{1.00} & \textbf{1.00} & \textbf{0.88} & \textbf{0.96}  \\ \hline
  \multirow{6}{*}{\makecell{{\#k @} \\ \cvratio>90\%}}   & Rand          & >100         & 7            & 10           & 34                & 8             & 10            & 94           & 22             \\ \cline{2-10}
                        & Degree        & >100         & 4            & 17           & 45                & 9             & 4             & 93           & 26             \\ \cline{2-10}
                        & ErrRate       & 71           & 2            & 7            & 18                & 3             & 3             & 19           & 13             \\ \cline{2-10}
                        & DegXErr       & 71           & 2            & 6            & 18                & 3             & 3             & 32           & 7              \\ \cline{2-10}
                        & Appx          & 61           & \textbf{1}   & 5            & 18                & 3             & 5             & 30           & 15             \\ \cline{2-10}
                        & Egal          & \textbf{55}  & \textbf{1}   & \textbf{3}   & \textbf{12}       & \textbf{1}    & \textbf{1}    & \textbf{9}   & \textbf{2}     \\ \hline
    \multirow{6}{*}{\makecell{{\#k @} \\ \cvratio>75\%}} & Rand          & 83           & 8            & 16           & 61                & 15            & 14            & 37           & 55             \\ \cline{2-10}
                        & Degree        & 83           & 5            & 28           & 64                & 14            & 6             & 20           & 54             \\ \cline{2-10}
                        & ErrRate       & 46           & 3            & 10           & 31                & 5             & 4             & 14           & 26             \\ \cline{2-10}
                        & DegXErr       & 51           & 3            & 8            & 36                & 4             & 3             & 8            & 16             \\ \cline{2-10}
                        & Appx          & 47           & \textbf{2}   & 9            & 35                & 6             & 7             & 11           & 39             \\ \cline{2-10}
                        & Egal          & \textbf{36}  & \textbf{2}   & \textbf{4}   & \textbf{19}       & \textbf{2}    & \textbf{2}    & \textbf{4}   & \textbf{3}     \\ \hline
  \end{tabular}
\end{footnotesize}
\end{table*}

\vspace{-0.3cm}
\paragraph{Real-World Graphs}
We also evaluate our methods on four diverse real network datasets~\cite{NetworkDataRepository}, BIO~\cite{BIOnet1,BIOnet2}, CSPK~\cite{DIMACS10_delaunay}, FB~\cite{foodnetwork}, WIKI~\cite{WIKInetwork}. 
Here we also apply finite step FJ model to construct weight matrix $\lW$ and  randomly sample  $\mathbf{\hy}$ from  Bernoulli distributions. 
\begin{figure}[ht]
  \centering
    \includegraphics[width=0.45\linewidth]{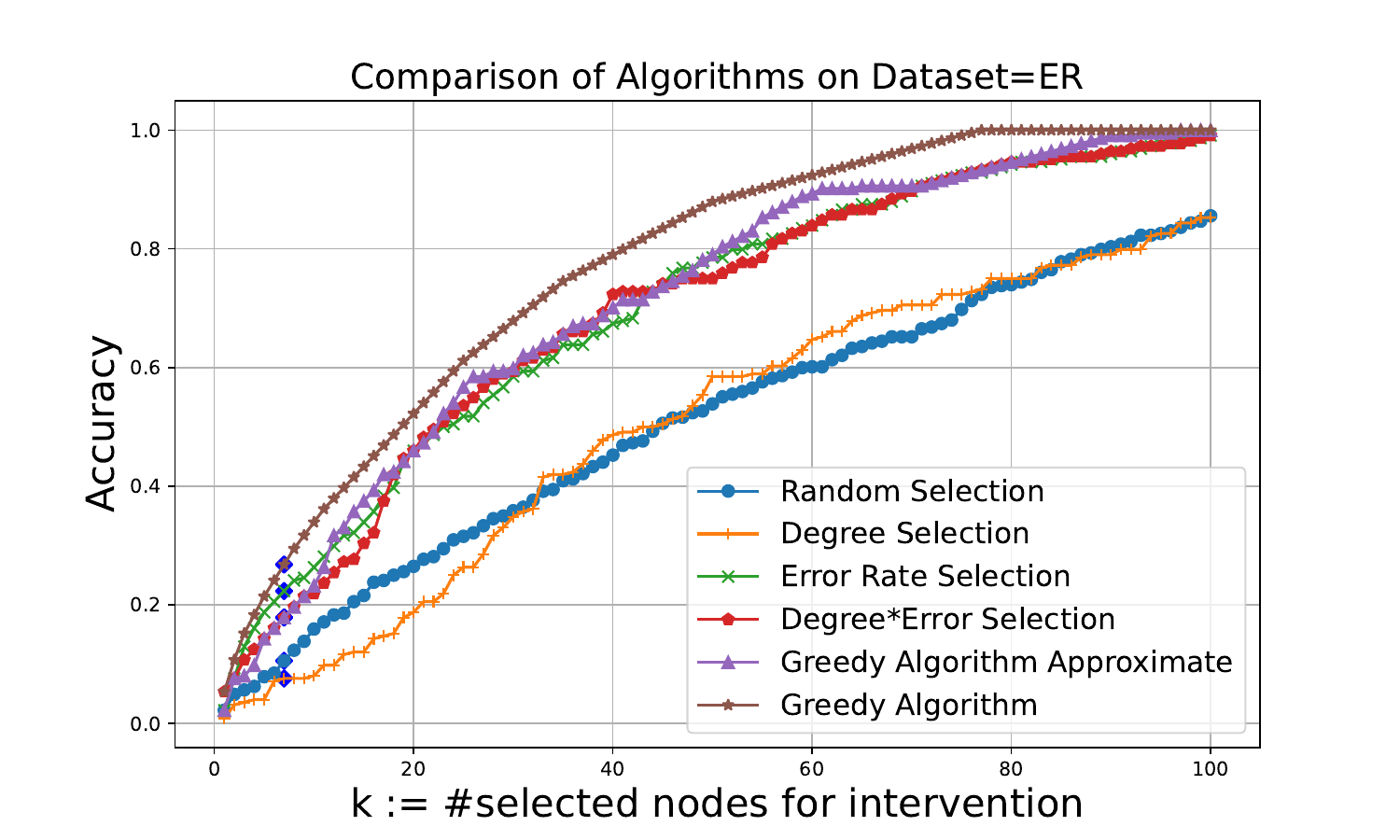}
\includegraphics[width=0.45\linewidth]{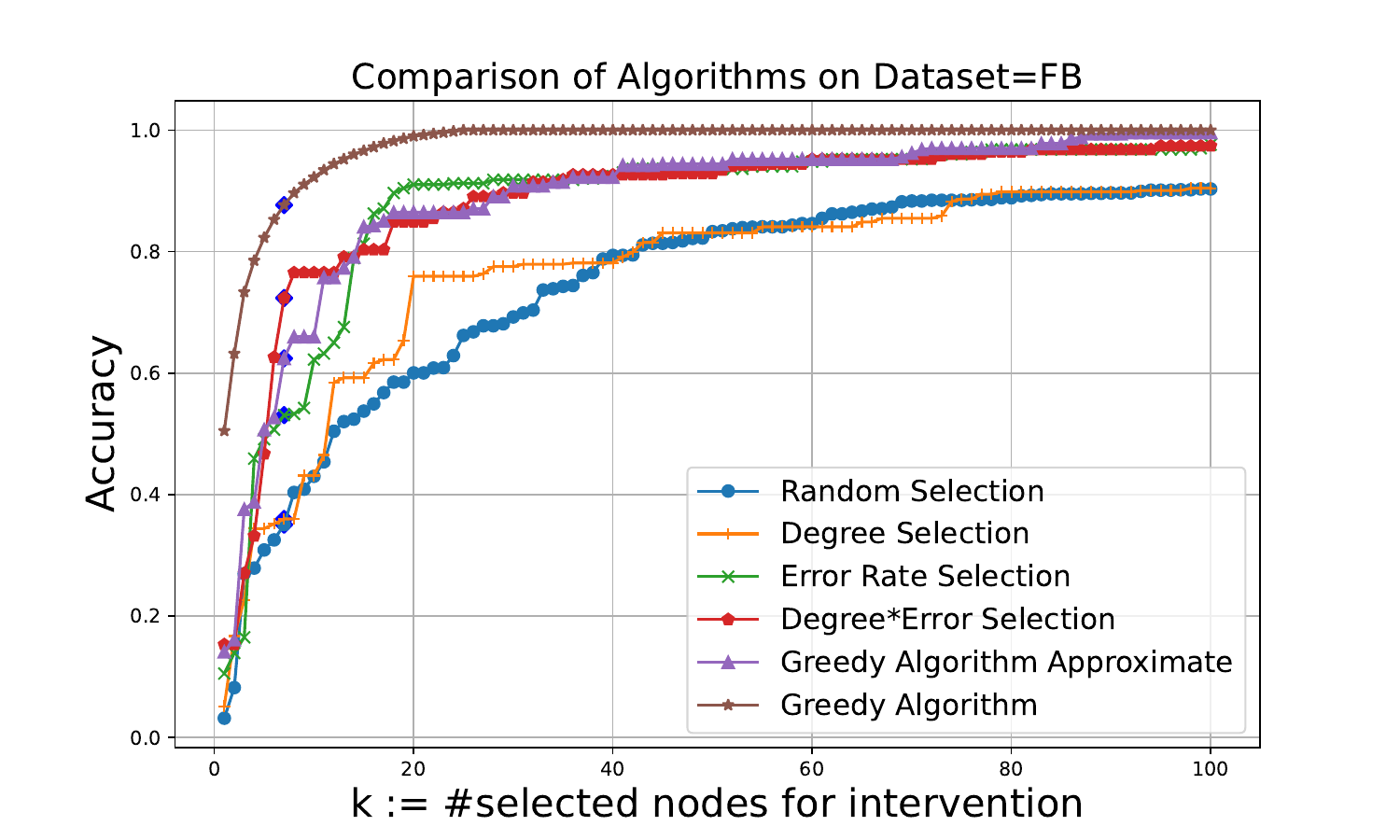}
    \vspace{-8pt} 
  \caption{Algorithms performance on ER (top) and WIKI (bottom).}
  \label{fig:exp} 
\end{figure}

We define our accuracy \cvratio\ to be the achieved egalitarian gain, normalized by its upper bound. For each dataset, we progressively increase  $k$. Obviously, 
as $k$ increases, the accuracy should increase. Therefore, we fix the threshold value $k=\log(n)$ for different datasets and compare the corresponding \cvratio\ of different methods. We also report the number of modified nodes $k$ required to achieve certain levels of accuracy. See \cref{tb:exp} for details. Figure~\ref{fig:exp} shows two of these results on real and synthetic graphs and more are presented in Figure~\ref{fig:exp_appendix}. 
More details
about experiments can be found in~\cref{asec:exp_details}. Code of our experiments is available through link \footnote{\tiny{\url{https://github.com/jackal092927/social_learning_public}}}.


From the results we can conclude that, in summary, all six algorithms can be categorized into four tiers: Tier1=\{\algo\}, Tier2=\{\appalgo\}, Tier3=\{DegXErr, ErrRate\}, Tier4=\{Degree, Rand\}.
We rank the efficiency of these methods as: Tier1$\gg$Tier2$\geq$Tier3$\gg$Tier4.
Our greedy algorithm \algo\ in general performs best on all datasets. In some datasets the \appalgo\ algorithm outperforms algorithms in Tier3 \& 4 when $k$ is very small, however, the performances of Tier2 \& 3 algorithms quickly become similar as $k$ increases. Tier4 algorithms are always the slowest in improving our egalitarian objective function.
On almost all datasets, our greedy algorithm can achieve more than 80\% accuracy within only $\log n$ nodes selected to intervene, and it beats all the other baselines.
Our greedy approximation algorithm can achieve more than 70\% accuracy. On all datasets, it beats our two baselines in Tier4 and on some data sets it beats all the baselines of Tier3 \& 4.

\vspace{-0.2cm}
\section*{Conclusion}

Given a network in which agents cooperatively perform a classification task, we analyse the problem of optimally choosing $k$ vertices and improving their
innate predictions to maximize the overall network improvement.

\vspace{-0.3 cm}
\paragraph{Limitations and Future Directions}
In this paper, our modeling relies on a few  simplifications which may pose limitations in the applicability of the methods, and they may be addressed in future works:
\vspace{-0.2 cm}
\begin{enumerate}
    \item We assume that the social planner is capable of improving every agent's innate prediction equally through \cref{eq:improvementphi}. In reality, this improvement may depend on the \emph{agent}, $i$, as well as the \emph{data point}, $a$. 

    \item Our analyses are valid when the social influence graph has non-negative weights and, in the current form, they do not generalize to graphs with negative weights, e.g.,   signed graphs. 
\end{enumerate}

We believe that  overcoming any of the above limitations  would be an interesting extension of our work, and we propose them as future directions. 



\section*{Acknowledgements} Xin and Gao would like to acknowledge NSF support through CCF-2118953, CCF-2208663, DMS-2311064, DMS-2220271, and IIS-2229876.


\bibliography{main}
\bibliographystyle{icml2024}

\newpage
\appendix
\onecolumn
\section{Additional material: proofs and  details of algorithms}

Let $\Omega^+=\{a\in \Omega\mid \y(a)=+1\}$ and  $\Omega^-=\{a\in \Omega\mid \y(a)=-1\}$. The following are some useful lemmas that we will use throughout:

\begin{lemma}\label{lem:exerr}
Let $v_j\in V$ be an arbitrary agent.  We have: 
\[
\Ex_{a\sim \Omega^+}\left[\hy_i(a)\right]= 1-2\err(v_i), \quad \Ex_{a\sim \Omega^-}\left[\hy_i(a)\right]= 2\err(v_i)-1~.
\]
\end{lemma}
\begin{proof}
\begin{align*}
 \Ex_{a\sim \Omega^+}\left[\hy_i(a)\right]&= +1 \cdot \mathbb{P}\left(\hy_i(a)=+1\right)-1\cdot \mathbb{P}\left(\hy_i(a)=-1\right)\\
 &= +1 \cdot \mathbb{P}(\hy_i(a)=\y(a))-1\cdot \mathbb{P}(\hy_i(a)\neq \y(a))\\
 &= \left(1-\err(v_i)\right)-\err(v_i)\\
 &= 1- 2\err(v_i)~.
\end{align*}
Similarly we have that 
\begin{align*}
 \Ex_{a\sim \Omega^-}\left[\hy_i(a)\right]&= +1 \cdot \mathbb{P}\left(\hy_i(a)\neq \y(a)\right)-1    \cdot \mathbb{P}\left(\hy_i(a)=\y(a)\right)= \err(v_i)-\left(1-\err(v_i)\right)=  2\err(v_i)-1~.
\end{align*}
\end{proof}

\begin{proposition}\label{prop:positive_guarantee}
Given any  non-negative $\lW$ used in equation~\eqref{eq:finite} and  $S$ on which we improve $\hy_j$ to $\ty_j$, we have
\[\forall v_i\in V,\  \Bnew(i,a)>\B(i,a) \iff \exists v_j\in S, \y(a)\hy_j(a)=-1 \wedge \lW_{ij}>0~.\]

\end{proposition}
\begin{proof}
By looking at \cref{eq:intervention}, it is evident that if $\hy_j(a)=\y(a)$ then $\ty_j(a)=\hy_j(a)$. Thus, improving $v_j$ will only improve prediction on $a\in \Omega$ iff $\hy_j(a)\neq \y(a)$ or equivalently $\hy_j(a) \y(a)=-1$. Note that for any other $v_i\in V$, $\ty_j(a)$ appears in $\Bnew(i,a)$ with coefficient $\lW_{ij}$. 
Since the matrix is  non-negative  we either have  $\Bnew(i,a)=\B(i,a)$  or
$\Bnew(i,a)>\B(i,a)$. 
Therefore, we will have 
  $\Bnew(i,a)>\B(i,a)$ iff  $\hy_j(a) \y(a)=-1$ and  $\lW_{ij}\neq 0$.
\end{proof}


\subsection{Missing proofs from \cref{sec:algo:agg}: Analysis of aggregate optimization}\label{app:algo:agg}

\begin{proof}[\textbf{Proof of \cref{lem:agg1}}]

Let $\Omega^+=\{a\in \Omega\mid \y(a)=+1\}$ and  $\Omega^-=\{a\in \Omega\mid \y(a)=-1\}$.
We have that:

\begin{align}\label{eq:11}
       \Gagg(S)=& 
\Ex_{a\sim\Omega}\left[ \sum_{i=1}^n \Bnew(i,a)- \B(i,a)
\right]\nonumber\\
=& \Ex_{a\sim\Omega^+}\left[ \sum_{i=1}^n \znew^*(i,a)-\z^*(i,a)
\right]\mathbb{P}(a\in \Omega^+)+\Ex_{a\sim\Omega^-}\left[ \sum_{i=1}^n \z^*(i,a)-\znew^*(i,a)
\right]\mathbb{P}(a\in \Omega^-)
\nonumber \\
=&\sum_{i=1}^n \left(\Ex_{a\sim\Omega^+}\left[  \znew^*(i,a)-\z^*(i,a)
\right]\mathbb{P}(a\in \Omega^+)+\Ex_{a\sim\Omega^-}\left[  \z^*(i,a)-\znew^*(i,a)
\right]\mathbb{P}(a\in \Omega^-)\right)~.
\end{align}

Note that:
\[
\z^*(i,a)-\znew^*(i,a)= \sum_{j=1}^n \lW_{ij}(\hy_j(a)-\ty_j(a))= \sum_{j\in S} \lW_{ij}(\hy_j(a)-\ty_j(a))= \imp \sum_{j\in S}\lW_{ij}[\hy_j(a)-\y(a)]
\]

Thus, 

$$\z^*(i,a)-\znew^*(i,a)= 
\begin{cases}
\imp \sum_{j\in S}\lW_{ij}[\hy_j(a)+1]& \text{if }a\in \Omega^-\\
\imp \sum_{j\in S}\lW_{ij}[\hy_j(a)-1]& \text{if }a\in \Omega^+
\end{cases}
$$
 Using linearity of expectation and plugging in \cref{lem:exerr} we have:

 $$
 \Ex\left[\z^*(i,a)-\znew^*(i,a)\right]= 
\begin{cases}
\imp \sum_{j\in S}\lW_{ij}[2\err^{(0)}(v_i)-1+1]& \text{if }a\in \Omega^-\\
\imp \sum_{j\in S}\lW_{ij}[1-2\err^{(0)}(v_i) -1]& \text{if }a\in \Omega^+
\end{cases}
$$
Therefore, 

 $$
 \Ex_{a\sim \Omega^-}\left[\z^*(i,a)-\znew^*(i,a)\right]= 
2\imp\ \err^{(0)}(v_i) \sum_{j\in S}\lW_{ij}=
 \Ex_{a\sim \Omega^+}\left[\znew^*(i,a)- \z^*(i,a)\right]~.
$$

Plugging in the above in  \cref{eq:11} we obtain:
\[
    \Gagg(S)=2\imp \sum_{j\in S}\sum_{i=1}^n \lW_{ij} \err^{(0)}(v_i)[\mathbb{P}(a\in \Omega^+)+\mathbb{P}(a\in \Omega^-)]=2\imp \sum_{j\in S}\sum_{i=1}^n \lW_{ij} \err^{(0)}(v_i)~.
\]

This means by picking $k$ vertices  with highest values of ${\rm Inf}(j)=\sum_{i=1}^n \lW_{ij} \err^{(0)}(v_i)$ we will obtain the optimal solution for \cref{problem1}.
In order to find these values, we first need to find all the values of ${\rm Inf}(j)$ for all $j\in V$. Which takes $\Theta(n^2)$ number of steps. Then we have to find top $k$ elements among these values, which will take $\Theta(kn)$.
\end{proof}

\begin{proof}[\textbf{Proof of \cref{thm:agg2}}]
    Let $\widehat{\rm Inf}(j)=\sum_{i=1}^n\hW_{i,j}\err^{(0)}(v_j)$. For all $j$, we have that 
    \[
 \left|
    {\rm Inf}(j)- \widehat{\rm Inf}(j)
    \right|= \sum_{i=1}^n  \left|
   \lW_{ij}- \hW_{ij}
    \right|\err^{(0)}(v_j)\leq
     \sum_{i=1}^n  \ \left|
   \lW_{ij}- \hW_{ij}
    \right|
    \]
    Note that the right-hand side is the $\ell_1$ norm of the $j$th column of $   \lW- \hW$, let's denote the $j$th column of these matrix respectively by $\lW_{\cdot j}$ and $\hW_{\cdot j}$. Since the $\ell_1$  norm of a matrix is defined the be the maximum over $\ell_1$ norm of all of its columns we have that
    \[\forall j \ 
     \left|
    {\rm Inf}(j)- \widehat{\rm Inf}(j)
    \right|\leq 
    \left|
   \lW_{\cdot j}- \hW_{\cdot j}
    \right|_1\leq \epsilon ~.\]

In the previous theorem we showed that 
\[
\Gagg(S)= 2\imp \sum_{j\in S} {\rm Inf}(j)~.
\]

Since size of $S$ is $k$ and $\imp\leq 1$ the total error is bounded by $2k \epsilon $~.
\end{proof}

\subsection{Missing proof from \cref{sec:algo:egal}: Analysis of egalitarian optimization}

\subsubsection{\textbf{Hardness Results}}

\begin{theorem}\label{app:thm:hardness}
    There is a polynomial time reduction from the set cover problem to \cref{problem1}.
\end{theorem}
\begin{proof}
Consider an arbitrary $S\subseteq V$.
    We use \cref{prop:positive_guarantee}
to simplify $\Gegal(S)$. 
For any $v_i\in V$, we define:
$\Omega_{v_j}\triangleq \{a\in \Omega \mid y(a)\hy_j(a)=-1\}$. 
For any $a\in \Omega$, we denote its probability by $\mu_a$.
For a given $S$ we have: 

\begin{align}\label{eq:12}
\Gegal(S) &= \Ex_{a\sim\Omega} \left[\sum_{i=1}^n \mathbf{1}(\B(i,a)<0 \wedge \B(i,a)< \Bnew(i,a))\right]\nonumber \\
&=\sum_{a\in \Omega}\sum_{i=1}^n \mu_a \mathbf{1}(\B(i,a)<0)\cdot\mathbf{1}\left(\bigvee_{v_j\in S}\left( 
a\in \Omega_{v_j} \wedge \lW_{ij}>0
\right)\right)\nonumber\\
&= \sum_{(a,v_i)\in \Omega\times V}\mu_a\mathbf{1}(\B(i,a)<0)\cdot\mathbf{1}\left(\bigvee_{v_j\in S}\left( a\in \Omega_{v_j} \wedge \lW_{ij}>0 \right)\right)
\end{align}

Using the above simplification, we now construct the following instance of the weighted set cover problem:

Consider a bipartite graph where one part is  $U=\{(v_i,a)\in V\times \Omega \mid  \B(i,a)<0\}$ and the other part is $S=V$. 
There is an edge between any $v_j\in V$ to a pair $(v_i,a)\in U$ iff $a\in \Omega_{v_j} \wedge \lW_{ij}>0$ 
and the weight of each pair $(a,v_i)$ is $\mu_a$.
Under the new definition \cref{eq:12} is equivalent to 

\[
\Gegal(S)= \sum_{(a,v_i)\in U} \mu_a \cdot \mathbf{1}\left(\bigvee_{v_j\in S}\left( a\in \Omega_{v_j} \wedge \lW_{ij}>0 \right)\right)
\]

Note that 
\[\mathbf{1}\left(\bigvee_{v_j\in S}\left( a\in \Omega_{v_j} \wedge \lW_{ij}>0 \right)\right)=1 \iff \text{there is an edge between $v_j$ and $(a,v_i)$ and $v_j\in S$}~.\]
 The reduction is polynomial in sizes of  $\Omega$ and $V$. Since we assume that $\Omega$ has polynomial size, it is a polynomial time reduction.  This completes the proof. 
\end{proof}

\begin{proof}[\textbf{Proof of \cref{thm:hardness}}]
The proof follows from \cref{app:thm:hardness} and the fact that set cover is NP-hard. 
\end{proof}

\begin{theorem}\label{app:thm:hard:err}Consider \cref{problem2} and assume $k=1$.
    There exist two networks with the same number of agents, same $\lW$ and same error rates $\{\err(v_j)\}_{j=1}^n$. In these network, only the joint probability distributions $\pi_1$ and $\pi_2$ are different. 
    There are subsets $V_1,V_2\subseteq V$ such that $V_1\cap V_2=\emptyset$. 
    In the first network we have that for any  $u\in V_1, \Gegal(u)=\Theta(n)$ and for any $u\notin V_1$ we have $\Gegal(u)=\Theta(1)$. In the second network for any  $u\in V_2$ will have $\Gegal(u)=\Theta(n)$ and any $u\notin V_2$ will satisfy $\Gegal(u)=\Theta(1)$.

\end{theorem}
\begin{proof}
Let $V=\{u_1,u_2,u_3,u_4\}\cup\{v_1,v_2,v_3,\dots , v_{2n}\}$. We define $\lW$ to be the following matrix:

$\lW_{u_1,v_j}=\lW_{u_2,v_j}=1$ for all $j=1:n$ , and $\lW_{u_3,v_j}=\lW_{u_4,v_j}=1$ for all $j=n+1:2n$.

For each vertex in $V$ we also have a self loop of weight $1$, i.e., $\lW_{u_i u_i}=\lW_{v_jv_j}=1$ for all $i,j$.

The error rates of these agents are as follows:
$\err(v_j)=0$ for all $j=1:2n$ and $\err(u_i)=1/2$ for all $j=1:4$.

In the first network the error of $\hy_{u_1}(a)$ is negatively correlated with  $\hy_{u_2}(a)$ and $\hy_{u_3}$ is positively correlated with $\hy_{u_4}$  as: 
\[
\Prob\left(\hy_{u_1}(a)\neq \hy_{u_2}(a)\right)=1 \quad \& \ \Prob\left(\hy_{u_3}(a)= \hy_{u_4}(a)\right)=1 
\]
Both $\hy_{u_1}$ and $\hy_{u_2}$ are independent from $\hy_{u_3}$ and $\hy_{u_4}$.

Let $V_1=\{u_3,u_4\}$. We now show that $\Gegal(u_3)=\Gegal(u_4)=n$ and for any $u\notin V_1$, $\Gegal(u)\in\{0,1\}$.

From \cref{lem:greedy} we conclude that for any vertex $u$ we have:

\[
\Gegal(u)=\sum_{i; \lW_{iu}\neq 0}
\Prob_{a\sim \Omega}
\left(
\B(i,a)\leq 0  \wedge \y(a)\neq  \hy_u(a)
\right)
\]

Thus, it is immediate that for each $v_j$, we have $\Gegal(v_j)=0$. 

Consider $u_1$ 
\begin{align*}
    \Gegal(u_1)&=\sum_{i=1:n}
\Prob_{a\sim \Omega}
\left(
\B(i,a)\leq 0  \wedge \y(a)\neq  \hy_{u_1}(a)
\right)+\Prob_{a\sim \Omega}
\left(
\B(u_1,a)\leq 0  \wedge \y(a)\neq  \hy_{u_1}(a)
\right) 
\end{align*}
Since $u_1$ is connected to $v_1,\dots v_n$, the last summand is $0$ if $n>1$, and it is $1/2$ if $n=1$. In any case it is a constant. We now look at the first summand. 
\begin{align*}
first\  summand &= \sum_{i=1:n}
\Prob_{a\sim \Omega} \left(
(\hy_{u_1}(a)+\hy_{u_2}(a)+\hy_{i}(a) )\cdot \y(a)\leq 0 \wedge   \y(a)\neq  \hy_u(a)
\right)\\
&= \sum_{i=1:n}
\Prob_{a\sim \Omega} \left(
 \y(a)^2\leq 0 \wedge   \y(a)\neq  \hy_u(a)
\right)=0
\end{align*}
The last equation follows from the fact that always $\hy_{u_1}(a)\neq\hy_{u_2}(a)$, and $\hy_i(a)=\y(a)$.

Similarly we can show that $ \Gegal(u_2)=\{0,1/2\}$.

For $u_3$, let $c$ be a constant which is $c\in\{0,1/2\}$. We have:
\begin{align*}
    \Gegal(u_3)&=\sum_{i=n+1:2n}
\Prob_{a\sim \Omega}
\left(
\B(i,a)\leq 0  \wedge \y(a)\neq  \hy_{u_3}(a)
\right)+\Prob_{a\sim \Omega}
\left(
\B(u_3,a)\leq 0  \wedge \y(a)\neq  \hy_{u_3}(a)
\right) \\
&= \sum_{i=1:n}
\Prob_{a\sim \Omega} \left(
(\hy_{u_3}(a)+\hy_{u_4}(a)+\hy_{i}(a) )\cdot \y(a)\leq 0 \wedge   \y(a)\neq  \hy_u(a)
\right)+c\\
&= \sum_{i=1:n}
\Prob_{a\sim \Omega} \left(
 (\y(a)-2\y(a))\cdot\y(a) \leq 0 
\right)\err(u_3)+c\\
&=n/2+c~.
\end{align*}
Similarly we have $ \Gegal(u_4)\in \{n/2,(n+1)/2\}$.

Therefore, both $u_3$ and $u_4$ can be the optimal choice for this network. And any other choice will have an error of  magnitude $\Theta(n)$.

In the second network, we make the following change: 

\[
\Prob\left(\hy_{u_1}(a)= \hy_{u_2}(a)\right)=1 \quad \& \ \Prob\left(\hy_{u_3}(a)\neq  \hy_{u_4}(a)\right)=1 
\]
We still let both $\hy_{u_1}$ and $\hy_{u_2}$ are independent from $\hy_{u_3}$ and $\hy_{u_4}$.

Using a similar analysis we can show that 

\[
\Gegal(u_1)=\Gegal(u_2)=n/2+1 \ \& \ \forall u\in V, \ u\neq u_1,u_2\implies \Gegal(u)\in \{0,1/2\}~.
\]

\end{proof}

\subsubsection{\textbf{Monotonicity and submodularity of $\Gegal$}}

\begin{lemma}\label{app:lem:monotone}
Assume $S'\subseteq S\subseteq V$, we have: $\Gegal(S')\leq \Gegal(S)$~.
\end{lemma}
\begin{proof}
From \cref{prop:positive_guarantee} that for any arbitrary $S\subseteq V$ we have:

\[
\Gegal(S)=\sum_{i=1}^n\Prob
\left(
\y(a) z^*(i,a)\leq 0 \wedge \bigvee_{\substack{v_j\in S\\ \lW_{ij}\neq 0}} (  \y(a)\neq \hy_j(a))
\right)
\]

For $S'\subseteq S$, we split the event $\bigvee_{\substack{v_j\in S\\ \lW_{ij}\neq 0}} (  \y(a)\neq \hy_j(a))$ to the two following non-intersecting events:

\[
\bigvee_{\substack{v_j\in S\\ \lW_{ij}\neq 0}} (  \y(a)\neq \hy_j(a))= \underbrace{\left(
\bigvee_{\substack{v_j\in S'\\ \lW_{ij}\neq 0}} (  \y(a)\neq \hy_j(a)) 
\right)}_{E_{i1}}
\vee 
\underbrace{\left(
\bigvee_{\substack{v_j\in S\setminus S'\\ \lW_{ij}\neq 0}} (  \y(a)\neq \hy_j(a))\wedge  \bigwedge_{\substack{v_j\in S'\\ \lW_{ij}\neq 0}} (  \y(a)= \hy_j(a))
\right)}_{E_{i2}}
\]

Since $E_1$ and $E_2$ are non-intersecting we have: 
\[
\Gegal(S)= \sum_{i=1}^n \mathbb{P}(\y(a)\z^*(i,a)<0\wedge E_{i1})+ \sum_{i=1}^n \mathbb{P}(\y(a)\z^*(i,a)<0\wedge E_{i2})~.
\]

Note that $\Gegal(S')=\sum_{i=1}^n \mathbb{P}(\y(a)\z^*(i,a)<0\wedge E_{i1})$. Therefore, we conclude the premise. 
\end{proof}

\begin{lemma}\label{app:lem:submodular}
    Consider arbitrary  $S\subseteq V$ and  $u,v\in {V}\setminus S$. We have: 
\[\Gegal(S\cup\{u,v\})+\Gegal(S)\leq \Gegal(S\cup\{u\})+\Gegal(S\cup\{v\})~.\]
\end{lemma}
\begin{proof}
Like previous lemma, we split the events of RHS and LHS to non-intersecting smaller events. 

Consider the following events:
\begin{align*}
  E_{FFF}(i)\equiv 
\left(\bigvee_{\substack{v_j\in S\\\lW_{ij}\neq 0  }}\hy_i(a)\neq \y(a) \right)\wedge \left(\hy_u(a)\neq \hy(a) \wedge \lW_{iu}\neq 0 \right)\wedge \left(\hy_v(a)\neq \y(a) \wedge \lW_{iv}\neq 0 \right)\\
  E_{TFF}(i)\equiv 
\left(\neg \bigvee_{\substack{v_j\in S\\\lW_{ij}\neq 0  }}\hy_i(a)\neq \y(a) \right)\wedge (\hy_u(a)\neq \hy(a) \wedge \lW_{iu}\neq 0 )\wedge \left(\hy_v(a)\neq \y(a) \wedge \lW_{iv}\neq 0 \right)\\
  E_{TTF}(i)\equiv 
\left(\neg \bigvee_{\substack{v_j\in S\\\lW_{ij}\neq 0  }}\hy_i(a)\neq \y(a) \right)\wedge \neg \left(\hy_u(a)\neq \hy(a) \wedge \lW_{iu}\neq 0 \right)\wedge \left(\hy_v(a)\neq \y(a) \wedge \lW_{iv}\neq 0 \right)\\
  E_{TFT}(i)\equiv 
\left(\neg \bigvee_{\substack{v_j\in S\\\lW_{ij}\neq 0  }}\hy_i(a)\neq \y(a) \right)\wedge  \left(\hy_u(a)\neq \hy(a) \wedge \lW_{iu}\neq 0 \right)\wedge \neg \left(\hy_v(a)\neq \y(a) \wedge \lW_{iv}\neq 0 \right)
\end{align*}

and similar definitions for $E_{TTT}(i)$, $E_{FTT}(i)$, $E_{FTF}(i)$ and $E_{FFT}(i)$. 

\[
\Gegal(S\cup \{u,v\})=\sum_{i=1}^n \sum_{(X,Y,Z)\in \{T,F\}^3 \setminus \{(F,F,F)\}}\Prob\left(
\y(a)\z^*(i,a)\leq 0 \wedge E_{X,Y,Z}(i)
\right),
\]

\[
\Gegal(S )= \sum_{i=1}^n \sum_{(Y,Z)\in \{T,F\}^2 }\Prob\left(
\y(a)\z^*(i,a)\leq 0 \wedge E_{T,Y,Z}(i)
\right),
\]

\[
\Gegal(S \cup \{u\})= 
\sum_{i=1}^n
\sum_{(X,Y,Z)\in \{T,F\}^3\setminus \{(F,F,F),(F,F,T)\} }\Prob\left(
\y(a)\z^*(i,a)\leq 0 \wedge E_{X,Y,Z}(i)
\right)
\]
and finally 

\[
\Gegal(S \cup \{v\})= 
\sum_{i=1}^n
\sum_{(X,Y,Z)\in \{T,F\}^3\setminus \{(F,F,F),(F,T,F)\} }\Prob\left(
\y(a)\z^*(i,a)\leq 0 \wedge E_{X,Y,Z}(i)
\right)
\]

By counting  the number of appearances of each term on the RHS and LHS we may conclude the premise. 
\end{proof}

\subsection{Pseudocode of the greedy algorithms}\label{app:sec:psudocodes}
In this section we present our algorithms for egalitarian improvement. 

\paragraph{Overview of algorithms} The first algorithm  \algo\ has access to $\pi$ and  finds the greedy choice accurately. 

The second algorithm \appalgo\ 
receives 
parameters $mode$, $\vec{\boldsymbol{\err}}$ and $\lW$ as input parameters. If  we assume pairwise independence of classifiers, $mode=\rm ind$ and $\vec{\boldsymbol{\err}}$ contains agents' error rates.  
If we assume group dependency $mode=\rm gr$ and $\vec{\boldsymbol{\err}}$ contains the agent's individual error rates as well as $\err(\R)$ and $\err(\Bl)$. Depending on the mode of the algorithm, \appalgo calls subsequent procedures $\mathsf{EstGain}^{\rm ind}$ and  $\mathsf{EstGain}^{\rm gr}$ for estimating the greedy choice.

The pseudocodes are as follows and analysis is presented in subsequent subsections:

\begin{algorithm}
\begin{algorithmic}\caption{\algo$\left( \pi,\lW\right)$\label{alg:2c}}
\STATE $S=\emptyset$
\FOR{$i=1:k$}
\STATE $maxval=0$
\FOR{$j=1:n$}
\STATE $\Gain(S,v_j)=0$
\FOR{$a\in \Omega$}
\FOR{$\ell=1:n$}
\STATE $E=\B(i,a)\leq 0  \wedge \left(\bigwedge_{
\substack{v_j\in S
\\ 
\lW_{ji}\neq  0 
} } \y(a){=} \hy_j(a)\right)
\wedge \y(a){\neq} \hy_u(a)$
\IF{$E\equiv T $ and $\lW_{\ell j}\neq 0$ }
\STATE $\Gain(S,v_j)=\Gain(S,v_j)+\pi(E)$
\ENDIF
\ENDFOR
\ENDFOR
\STATE \textbf{if} ${\Gain}(S,v_j)\geq maxval$, 
\STATE \quad $g=v_j$
\ENDFOR 
\STATE $S=S\cup\{ g\}$
\ENDFOR 
\end{algorithmic}
\end{algorithm}

\begin{algorithm}
\caption{\appalgo$\left( mode, \vec{\boldsymbol{\err}},\lW\right)$\label{alg:2c}}
\begin{algorithmic}
\STATE $S=\emptyset$
\FOR{$i=1:k$}
\STATE $maxval=0$
\FOR{$j=1:n$}
\STATE $\widehat{\Gain}^{\rm indv}(S,v_j)=\mathsf{EstGain}^{\rm ind}(S,v_j,\{\err(v_i)\}_{i=1}^n, \lW)$ (\cref{proc:ind})
\STATE $\widehat{\Gain}(S,v_j)=\widehat{\Gain}^{\rm indv}(S,v_j)$
\IF{mode = {\rm gr}}
\STATE $\widehat{\Gain}^{\rm gr}(S,v_j)=\mathsf{EstGain}^{\rm gr}(S,v_j,\{\err(v_i)\}_{i=1}^n,\err(\R),\err(\Bl), \lW)$ (\cref{proc:gr})
\STATE $\widehat{\Gain}(S,v_j)=\rho \widehat{\Gain}^{\rm gr}(S,v_j)+(1- \rho) \widehat{\Gain}^{\rm indv}(S,v_j)$
\ENDIF
\STATE \textbf{if} $\widehat{\Gain}(S,v_j)\geq maxval$, 
\STATE \quad $g=v_j$
\ENDFOR 
\STATE $S=S\cup\{ g\}$
\ENDFOR 
\end{algorithmic}
\end{algorithm}
\begin{algorithm}\caption{
$\mathsf{EstGain}^{\rm ind}(S,u,\{\err(v_j)\}_{j=1}^n, \lW)$
}\label{proc:ind}
\begin{algorithmic}

 \STATE $\widehat{\Gain}(S,u)=0$
\FOR{$\ell=1:n$}
\STATE $\cor_\ell(S)=1$
\FOR{$v_m\in S$}
\STATE \textbf{if} $\lW_{\ell m}\neq 0$, \STATE \quad $\cor_\ell(S)=\cor_\ell(S)\times(1-\err(v_m))$
\ENDFOR
\STATE \textbf{if}\  $\val_\ell(S,u)<0$ and $\lW_{\ell u}\neq 0$
\STATE $\widehat{\Gain}(S,u)\ {+}{=}\  \err(u)\cdot \cor_\ell(S)$
\ENDFOR    
\end{algorithmic}
\end{algorithm}

\begin{algorithm}\caption{
$\mathsf{EstGain}^{\rm gr}(S,u,\{\err(v_j)\}_{j=1}^n ,\err(\R),\err(\Bl),\lW)$
}\label{proc:gr}
\begin{algorithmic}

 \STATE $\widehat{\Gain}(S,u)=0$
 \STATE \textbf{boolean} $Case1=T$
 \STATE \textbf{boolean} $Case2R=F$
  \STATE \textbf{boolean} $Case2B=F$
 \STATE \textbf{boolean} $Case3=F$
\FOR{$\ell=1:n$}
\STATE $\cor_\ell(S)=1$
\FOR{$v_m\in S$}
\IF{$\lW_{\ell m}\neq 0$} 
\STATE \textbf{if} $v_m\in \W$ \textbf{ then } $\cor_\ell(S)=\cor_\ell(S)\times(1-\err(v_m))$
\IF{($v_m\in \R )\wedge Case1=T$ }
\STATE $Case1= F$
\STATE $Case2R= T$
\ENDIF 
\IF{($v_m\in \Bl )\wedge Case1=T$ }
\STATE $Case1= F$
\STATE $Case2B= T$
\ENDIF 
\IF{$(v_m\in \R\wedge Case2B=T)$ or $(v_m\in \Bl\wedge Case2R=T)$  }
\STATE $Case2B=Case2R= F$
\STATE $Case3= T$
\ENDIF 
\ENDIF
\ENDFOR
\IF{$Case 1 \wedge u\in \W$}
\STATE $\Gain(S,u)+=\cor_\ell(S)\err(u)[\mathbf{1}(\val_i(\Bl\cup S),\R\cup\{u\})\err(\R)+\mathbf{1}(\R\cup S),\Bl\cup\{u\})\err(\Bl)]$
\ENDIF
\IF{$Case 1 \wedge u\in \R$}
\STATE $\Gain(S,u)+=\cor_\ell(S)\err(u)\mathbf{1}\val_i(\R\cup S,\Bl\cup\{u\})$
\ENDIF
\IF{$Case 1 \wedge u\in \Bl$}
\STATE $\Gain(S,u)+=\cor_\ell(S)\err(u)\mathbf{1}\val_i(\Bl\cup S,\R\cup\{u\})$
\ENDIF
\IF{$Case 2B \wedge u\in \W$}
\STATE $\Gain(S,u)+=\err(u)\err(\R)\cor(S)\mathbf{1}\val(\Bl\cup S,\R\cup \{u\})$
\ENDIF
\IF{$Case 2R \wedge u\in \W$}
\STATE $\Gain(S,u)+=\err(u)\err(\Bl)\cor(S)\mathbf{1}\val(\R\cup S,\Bl\cup \{u\})$
\ENDIF
\IF{$Case 2B \wedge u\in \R$}
\STATE   $\Gain(S,u)+=\err(\R)\cor(S)\mathbf{1}\val(\Bl\cup S,\R)$
\ENDIF
\IF{$Case 2R \wedge u\in \Bl$}
\STATE $\Gain(S,u)+=\err(\Bl)\cor(S)\mathbf{1}\val(\R\cup S,\Bl)$
\ENDIF
\ENDFOR    
\STATE \textbf{return} $\Gain(S,u)$
\end{algorithmic}
\end{algorithm}

\newpage

\subsubsection{\textbf{Finding the greedy choice}}\label{app:greedyfind}

Remember from the main text that 
\begin{equation*}\Gain_i(u,S)=
\Prob_{a\sim \Omega}
\left(
\B(i,a)\leq 0  \wedge \left(\bigwedge_{
\substack{v_j\in S
\\ 
\lW_{ji}\neq  0 
} } \y(a){=} \hy_j(a)\right)
\wedge \y(a){\neq} \hy_u(a)
\right)~.
\end{equation*}

\begin{proof}[\textbf{Proof of \Cref{lem:greedy}}]
It is immediate from \cref{prop:positive_guarantee} that for any arbitrary $S\subseteq V$ we have:

\[
\Gegal(S)=\sum_{i=1}^n\Prob
\left(
\y(a) z^*(i,a)\leq 0 \wedge \bigvee_{v_j\in S} (\lW_{ij}\neq 0 \wedge \y(a)\neq \hy_j(a))
\right)
\]
Writing the above for $S\cup \{u\}$ and simplifying we obtain:

\begin{align*}
    \Gegal(S\cup \{u\}) =& \sum_{i=1}^n\Prob
\left(
\y(a) z^*(i,a)\leq 0 \wedge \bigvee_{v_j\in S\cup \{u\}} (\lW_{ij}\neq 0 \wedge \y(a)\neq \hy_j(a))
\right)\\
=& \sum_{i=1}^n\Prob
\left(
\y(a) z^*(i,a)\leq 0 \wedge \bigvee_{v_j\in S} (\lW_{ij}\neq 0 \wedge \y(a)\neq \hy_j(a))
\right)\\
& ~ + \Prob
\left(
\y(a) z^*(i,a)\leq 0 \wedge \bigwedge_{v_j\in S} (\lW_{ij}= 0 \vee \y(a)= \hy_j(a))
\wedge (\lW_{iu}\neq 0\wedge \y(a)\neq \hy_u(a))
\right)
\\
=& \Gegal(S)+\sum_
{\substack{i=1:n
\\
\lW_{iu}\neq 0
} }\Prob
\left(
\y(a) z^*(i,a)\leq 0  \wedge \left(\bigwedge_{
\substack{v_j\in S
\\ 
\lW_{ij}\neq  0 
} } \y(a)= \hy_j(a)\right)
\wedge \y(a)\neq \hy_u(a)
\right)\\
=& \Gegal(S)+\sum_
{\substack{i=1:n
\\
\lW_{iu}\neq 0
} }\Gain_i(S,u)
\end{align*}
\end{proof}
The following lemma is a middle step for approximation of the greedy choice:

\begin{lemma}\label{app:middlestep}
Let $\Gain_i(S,u)$ be 
\[
\Gain_i(S,u)=
\Prob_{a\sim \Omega}
\left(
\B(i,a)\leq 0  \wedge \left(\bigwedge_{
\substack{v_j\in S
\\ 
\lW_{ij}\neq  0 
} } \y(a)= \hy_j(a)\right)
\wedge \y(a)\neq \hy_u(a)
\right)
\]

We have that 
\[\Gain_i(S,u)= 
\mathbb{P}\left(
T_i(S) \wedge F(u)
\right)
\Gamma_i(S,u)~,\]

where 
\begin{align*}
\Gamma_i(S,u)=&
\Prob\left(
a\in \Omega^+
\right)
\Prob_{a\in \Omega^+}\left(
\sum_{\substack{
j=1:n\\
v_j\notin S\cup\{u\}
}}\lW_{ij}\hy_j(a)< -\sum_{v_j\in S}\lW_{ji}+ \lW_{iu}\right)\\
&+
\Prob\left(
a\in \Omega^-
\right)
\Prob_{a\in \Omega^-}\left(
\sum_{\substack{
j=1:n\\
v_j\notin S\cup\{u\}
}}\lW_{ij}\hy_j(a)> \sum_{v_j\in S}\lW_{ji}- \lW_{iu}
\right)~.
\end{align*}
and $T_i(S)$ and $F(u)$ are the following events: 

\[
T_i(S)\doteq \bigwedge_{\substack{v_j\in S\\
\lW_{ij}\neq 0
}}\hy_j(a)=\y(a), \quad \quad F(u)\doteq \hy_u(a)\neq \y(a)
\]

\end{lemma}

\begin{proof}
Let $E$ denote the event 
\[E\triangleq\left(\bigwedge_{
\substack{v_j\in S
\\ 
\lW_{ji}\neq  0 
} } \y(a)= \hy_j(a)\right)
\wedge \y(a)\neq \hy_u(a)~.\]

We may write 
\[\Prob
\left(
\y(a) z^*(i,a)\leq 0  \wedge E
\right)=\Prob
\left(
\y(a) z^*(i,a)\leq 0  \mid E
\right)\Prob(E)\]

In order find $\Prob
\left(
\y(a) z^*(i,a)\leq 0 \mid  E
\right)$, we split the probability based on the true label of $a$:

If $a\in \Omega^+$ we have:

\[
\Prob\left(
\y(a)z^*(i,a)\leq 0 \mid E
\right)= \Prob\left(
z^*(i,a)\leq 0 \mid 
\left(\bigwedge_{
\substack{v_j\in S
\\ 
\lW_{ij}\neq  0 
} } \hy_j(a)=+1\right)
\wedge  \hy_u(a)=-1
\right)
\]

Note that 
\[z^*(i,a)=\sum_{j=1}^n \lW_{ij}\hy_j(a)\]
thus, 
\[
\Prob\left(
\y(a)z^*(i,a)\leq 0 \mid E
\right)=
\Prob\left(
\sum_{{\substack{j=1:n\\
v_j\notin S\cup\{u\}}} }\lW_{ij}\hy_j(a)+ \sum_{v_j\in S}\lW_{ij}- \lW_{iu}\leq 0
\right)
\]

Similarly, if $a\in \Omega^-$ we have:

\[
\Prob\left(
\y(a)z^*(i,a)\leq 0 \mid E
\right)= \Prob\left(
z^*(i,a)\geq 0 \mid 
\left(\bigwedge_{
\substack{v_j\in S
\\ 
\lW_{ij}\neq  0 
} } \hy_j(a)=-1\right)
\wedge  \hy_u(a)=+1
\right)
\]

Since 
\[z^*(i,a)=\sum_{j=1}^n \lW_{ij}\hy_j(a)\]
we have: 
\[
\Prob\left(
\y(a)z^*(i,a)\geq 0 \mid E
\right)=
\Prob\left(
\sum_{\substack{j=1:n\\
v_j\notin S\cup\{u\}}} \lW_{ij}\hy_j(a)- \sum_{v_j\in S}\lW_{ij}+ \lW_{iu}\geq 0
\right)
\]

Rearranging and putting together, we obtain the premise. 
    
\end{proof}

\subsection{Missing material from \cref{sec:algo:ind}: estimating greedy choice assuming independence}\label{app:sec:ind}

\begin{proof}[\textbf{Proof of \cref{lem:appgain}}]
Let 
\[
\cor_i(S)=\mathbb{P}\left(\bigwedge_{\substack{v_j\in S\\
\lW_{ij}\neq 0
}}\hy_j(a)=\y(a)\right)~.
\]
Using independence and 
from \cref{app:middlestep} we can write 
\begin{equation}\label{eq:gaingamma}
    \Gain_i(S,u)= \err(u) \cor_i(S)\Gamma_i(S,u)
\end{equation}

In \cref{approxGamma} which follows this proof, we show that by  taking 

 $$
    \widehat{\Gamma}_i(S,u)=
    \begin{cases}
    0& \text{ if } \val_i(S,u)\geq 0\\
    1& \text{ otherwise}
    \end{cases}
    $$

    we have 
    \[
    \abs{\widehat{\Gamma}_i(S,u)-\Gamma_i(S,u)}\leq  \exp\left(
    -\frac{\val_i(S,u)^2}{4\sum_{i=1}^n\lW_{ij}^2}
    \right)~\leq \exp\left(
    -\frac{\val_i(u)^2}{4\sum_{i=1}^n\lW_{ij}^2}
    \right),
    \] 

    Plugging in this approximation in \cref{eq:gaingamma} we obtain:

    \begin{equation}\label{eq:approxgamma}
    \widehat{\Gain_i}(S,u)= 
    \begin{cases}
            \err(u) \cor_i(S)& \text{ if }\val_i(S,u)<0\\
            0 & \text{ otherwise}
    \end{cases}
\end{equation}

Note that since the classifiers are independent we have  
\begin{align*}
    \cor_i(S)&=\mathbb{P}\left(\bigwedge_{\substack{v_j\in S\\
\lW_{ij}\neq 0
}}\hy_j(a)=\y(a)\right)= \prod_{\substack{v_j\in S\\
\lW_{ij}\neq 0
}} \mathbb{P}\left(\hy_j(a)=\y(a)\right)= \prod_{\substack{v_j\in S\\
\lW_{ij}\neq 0
}} \left(1-\err(v_j)\right)~.
\end{align*}

We have $  \cor_i(S),\err(u)\leq 1$, thus the error of approximating $\Gain_i(S,u)$ using \cref{eq:approxgamma} is at most  $ \exp\left(
    -\frac{\val_i(u)^2}{4\sum_{i=1}^n\lW_{ij}^2}
    \right)$~.
    
\end{proof}

\begin{lemma}\label{approxGamma}
    Assume that all the all the classifiers 
    $\hy_1(a), \hy_2(a),\dots, \hy_n(a) $ 
    are pairwise  independent. We can estimate $\Gamma_i(S,u)$ as follows:

    $$
    \widehat{\Gamma}_i(S,u)=
    \begin{cases}
    0& \text{ if } \val_i(S,u)>0\\
    1& \text{ otherwise}
    \end{cases}
    $$

    where, 
    \[
    \val_i(S,u)=\sum_{v_j\in S}\lW_{ij}+\sum_{\substack{
    j=1:n\\
    j\not \in S\cup\{u\} 
    }}\lW_{ij}[1-2\err(j)]-\lW_{iu}
    \]

    The error of this estimation is bounded by:

    \[
    \abs{\widehat{\Gamma}_i(S,u)-\Gamma_i(S,u)}\leq  \exp\left(
    -\frac{\val_i(S,u)^2}{4\sum_{i=1}^n\lW_{ij}^2}
    \right)~\leq \exp\left(
    -\frac{\val_i(u)^2}{4\sum_{i=1}^n\lW_{ij}^2}
    \right),
    \]    
    where 
    \[
    \val_i(u)= \sum_{j=1}^n \lW_{ij}[1-2\err(j)]-2\err(u)\lW_{iu}~.
    \]
\end{lemma}

\begin{proof}
Let's remember the definition of $\Gamma_i(S,u)$ from \cref{app:middlestep}:
\begin{align}\label{gammaeq}
\Gamma_i(S,u)=&
\Prob\left(
a\in \Omega^+
\right)
\Prob_{a\in \Omega^+}\left(
\sum_{\substack{
j=1:n\\
v_j\notin S\cup\{u\}
}}\lW_{ij}\hy_j(a)< -\sum_{v_j\in S}\lW_{ij}+ \lW_{iu}\right)\nonumber\\
&+
\Prob\left(
a\in \Omega^-
\right)
\Prob_{a\in \Omega^-}\left(
\sum_{\substack{
j=1:n\\
v_j\notin S\cup\{u\}
}}\lW_{ij}\hy_j(a)> \sum_{v_j\in S}\lW_{ij}- \lW_{iu}
\right)~.
\end{align}

Assume first that $\val_i(S,u)>0$. We  use the Hoeffding bound \cref{thm:hoeffding} to estimate

\begin{equation}\label{eq:twoprobs}
 \Prob_{a\in \Omega^+}\left(
\sum_{\substack{
j=1:n\\
v_j\notin S\cup\{u\}
}}\lW_{ij}\hy_j(a)< -\sum_{v_j\in S}\lW_{ij}+ \lW_{iu}\right)
\quad \& \ 
\Prob_{a\in \Omega^-}\left(
\sum_{\substack{
j=1:n\\
v_j\notin S\cup\{u\}
}}\lW_{ij}\hy_j(a)> \sum_{v_j\in S}\lW_{ij}- \lW_{iu}
\right)   
\end{equation}

In order to bound the first probability in  \cref{eq:twoprobs}, note that  $\sum_{\substack{
j=1:n\\
v_j\notin S\cup\{u\}
}}\lW_{ij}\hy_j(a)< -\sum_{v_j\in S}\lW_{ij}+ \lW_{iu}$ iff :

\begin{align*}  
\sum_{\substack{
j=1:n\\
v_j\notin S\cup\{u\}
}}\lW_{ij}\hy_j(a)<
\Ex\left[\sum_{\substack{
j=1:n\\
v_j\notin S\cup\{u\}
}}\lW_{ij}\hy_j(a)\right]- \Ex\left[\sum_{\substack{
j=1:n\\
v_j\notin S\cup\{u\}
}}\lW_{ij}\hy_j(a)\right]
-\sum_{v_j\in S}\lW_{ij}+ \lW_{iu}
\end{align*}

furthermore, from \cref{lem:exerr} we have $a\in \Omega^+$ implies:
\[
\Ex\left[
\sum_{\substack{
j=1:n\\
v_j\notin S\cup\{u\}
}}\lW_{ij}\hy_j(a)\right]= \sum_{\substack{
j=1:n\\
v_j\notin S\cup\{u\}
}}\lW_{ij} [1-2\err(j)]~.
\]

Thus, 
\begin{align*}
& \Prob_{a\in \Omega^+}\left(
\sum_{\substack{
j=1:n\\
v_j\notin S\cup\{u\}
}}\lW_{ij}\hy_j(a)< -\sum_{v_j\in S}\lW_{ij}+ \lW_{ui}\right)\\
=&\Prob_{a\in \Omega^+}\left(
\sum_{\substack{
j=1:n\\
v_j\notin S\cup\{u\}
}}\lW_{ij}\hy_j(a)<
\Ex\left[\sum_{\substack{
j=1:n\\
v_j\notin S\cup\{u\}
}}\lW_{ij}\hy_j(a)\right]- 
\sum_{\substack{
j=1:n\\
v_j\notin S\cup\{u\}
}}\lW_{ij} [1-2\err(j)]
-\sum_{v_j\in S}\lW_{ij}+ \lW_{iu}
\right)  \\
=&\Prob_{a\in \Omega^+}\left(
\sum_{\substack{
j=1:n\\
v_j\notin S\cup\{u\}
}}\lW_{ij}\hy_j(a)<
\Ex\left[\sum_{\substack{
j=1:n\\
v_j\notin S\cup\{u\}
}}\lW_{ij}\hy_j(a)\right]- 
\val_i(S,u)
\right)  \\
\end{align*}
Since $\hy_i(a)$s are pairwise independent, and $\val_i(S,u)>0$ we may use the Hoeffding bound to obtain:
\begin{align*}
     \Prob_{a\in \Omega^+}\left(
\sum_{\substack{
j=1:n\\
v_j\notin S\cup\{u\}
}}\lW_{ij}\hy_j(a)< -\sum_{v_j\in S}\lW_{ij}+ \lW_{iu}\right)
\leq \exp\left(-\frac{\val_i(S,u)^2}{\sum_{\substack{
j=1:n\\
v_j\notin S\cup\{u\}
}}\lW_{ij}^2}\right).
\end{align*}

The second probability in \cref{eq:twoprobs} may be bounded similarly as follows:
\begin{align*}
& \Prob_{a\in \Omega^+}\left(
\sum_{\substack{
j=1:n\\
v_j\notin S\cup\{u\}
}}\lW_{ji}\hy_j(a)> \sum_{v_j\in S}\lW_{ji}- \lW_{ui}\right)\\
=&\Prob_{a\in \Omega^+}\left(
\sum_{\substack{
j=1:n\\
v_j\notin S\cup\{u\}
}}\lW_{ji}\hy_j(a)>
\Ex\left[\sum_{\substack{
j=1:n\\
v_j\notin S\cup\{u\}
}}\lW_{ji}\hy_j(a)\right]- 
\sum_{\substack{
j=1:n\\
v_j\notin S\cup\{u\}
}}\lW_{ji} [2\err(j)-1]
+\sum_{v_j\in S}\lW_{ji}- \lW_{ui}
\right)  \\
=&\Prob_{a\in \Omega^+}\left(
\sum_{\substack{
j=1:n\\
v_j\notin S\cup\{u\}
}}\lW_{ji}\hy_j(a)>
\Ex\left[\sum_{\substack{
j=1:n\\
v_j\notin S\cup\{u\}
}}\lW_{ji}\hy_j(a)\right]+ 
\val_i(S,u)
\right)  \\
\end{align*}
Again, under pairwise independence and $\val_i(S,u)>0$ the above probability is bounded as:
\begin{align*}
     \Prob_{a\in \Omega^-}\left(
\sum_{\substack{
j=1:n\\
v_j\notin S\cup\{u\}
}}\lW_{ji}\hy_j(a)> \sum_{v_j\in S}\lW_{ji}- \lW_{ui}\right)
\leq \exp\left(-\frac{\val_i(S,u)^2}{\sum_{\substack{
j=1:n\\
v_j\notin S\cup\{u\}
}}\lW_{ji}^2}\right).
\end{align*}

Putting together, we obtain:
If $\val_i(S,u)>0$:
\begin{align*}
    \Gamma_i(S,u)=&
\Prob\left(
a\in \Omega^+
\right)
\Prob_{a\in \Omega^+}\left(
\sum_{\substack{
j=1:n\\
v_j\notin S\cup\{u\}
}}\lW_{ji}\hy_j(a)< -\sum_{v_j\in S}\lW_{ji}+ \lW_{ui}\right)\\
&\ \ \ +
\Prob\left(
a\in \Omega^-
\right)
\Prob_{a\in \Omega^-}\left(
\sum_{\substack{
j=1:n\\
v_j\notin S\cup\{u\}
}}\lW_{ji}\hy_j(a)> \sum_{v_j\in S}\lW_{ji}- \lW_{ui}
\right)~\\
\leq &   \exp\left(\frac{\val_i(S,u)^2}{\sum_{\substack{
j=1:n\\
v_j\notin S\cup\{u\}
}}\lW_{ji}^2}\right)[\Prob(a\in \Omega^+)+\Prob(a\in \Omega^-)]=\exp\left(\frac{\val_i(S,u)^2}{\sum_{\substack{
j=1:n\\
v_j\notin S\cup\{u\}
}}\lW_{ji}^2}\right) 
\end{align*}

Note that $\Gamma_i(S,u)\geq 0$. Therefore, if $\val_i(S,u)>0$, we define $\widehat{\Gamma}_i(S,u)=0$ and we will have:

\[0\leq  \Gamma_i(S,u)-\widehat{\Gamma}_i(S,u)\leq \exp\left(-\frac{\val_i(S,u)^2}{\sum_{\substack{
j=1:n\\
v_j\notin S\cup\{u\}
}}\lW_{ji}^2}\right) ~.\]

Let's now find a lower bound on $\frac{\val_i(S,u)^2}{\sum_{\substack{
j=1:n\\
v_j\notin S\cup\{u\}
}}\lW_{ji}^2}$ which is independent of $S$. We have:
\begin{align*}
\sum_{\substack{
j=1:n\\
v_j\notin S\cup\{u\}
}}\lW_{ji}^2\leq \sum_{i=1}^n
\lW_{ji}^2~.
\end{align*}
and 
\begin{align*}
\val_i(S,u)&= \sum_{v_j\in S}\lW_{ij}+\sum_{\substack{
    j=1:n\\
    j\not \in S\cup\{u\} 
    }}\lW_{ij}[1-2\err(j)]-\lW_{iu}\\
    &\geq \sum_{\substack{
    j=1:n\\
    j\neq u 
    }}\lW_{ij}[1-2\err(j)]-\lW_{iu}\\
   & = \sum_{j=1}^n \lW_{ij}[1-2\err(j)]- 2\err(u)\lW_{iu}~= \val_i(u)~.
\end{align*}
Thus,
\[0\leq  \Gamma_i(S,u)-\widehat{\Gamma}_i(S,u)\leq \exp\left(-\frac{\val_i(S,u)^2}{\sum_{\substack{
j=1:n\\
v_j\notin S\cup\{u\}
}}\lW_{ji}^2}\right)\leq \exp\left(-\frac{\val_i(u)^2}{\sum_{j=1}^n\lW_{ji}^2}\right) ~.\]

Assume now that $\val_i(S,u)< 0$.  In this case we write the first probability in \cref{eq:twoprobs} as:
\begin{align*}
& \Prob_{a\in \Omega^+}\left(
\sum_{\substack{
j=1:n\\
v_j\notin S\cup\{u\}
}}\lW_{ij}\hy_j(a)< -\sum_{v_j\in S}\lW_{ij}+ \lW_{iu}\right)
=  1-  \Prob_{a\in \Omega^+}\left(
\sum_{\substack{
j=1:n\\
v_j\notin S\cup\{u\}
}}\lW_{ij}\hy_j(a)\geq  -\sum_{v_j\in S}\lW_{ij}+ \lW_{iu}\right)
\\
= &1-\Prob_{a\in \Omega^+}\left(
\sum_{\substack{
j=1:n\\
v_j\notin S\cup\{u\}
}}\lW_{ij}\hy_j(a)\geq 
\Ex\left[\sum_{\substack{
j=1:n\\
v_j\notin S\cup\{u\}
}}\lW_{ij}\hy_j(a)\right]- 
\sum_{\substack{
j=1:n\\
v_j\notin S\cup\{u\}
}}\lW_{ij} [1-2\err(j)]
-\sum_{v_j\in S}\lW_{ji}+ \lW_{iu}
\right)  \\
=&1- \Prob_{a\in \Omega^+}\left(
\sum_{\substack{
j=1:n\\
v_j\notin S\cup\{u\}
}}\lW_{ij}\hy_j(a)\geq 
\Ex\left[\sum_{\substack{
j=1:n\\
v_j\notin S\cup\{u\}
}}\lW_{ij}\hy_j(a)\right]- 
\val_i(S,u)
\right)  
\end{align*}

Since $\val_i(S,u)< 0$, we have $-\val_i(S,u)> 0$ using the pairwise independence of the classifiers, we employ  the Hoeffding bound and obtain that:

\[
\Prob_{a\in \Omega^+}\left(
\sum_{\substack{
j=1:n\\
v_j\notin S\cup\{u\}
}}\lW_{ij}\hy_j(a)< -\sum_{v_j\in S}\lW_{ij}+ \lW_{iu}\right)
\geq 1- \exp\left(-\frac{\val_i(S,u)^2}{\sum_{\substack{
j=1:n\\
v_j\notin S\cup\{u\}
}}\lW_{ij}^2}\right)~.
\]

Similarly we have :

\[
\Prob_{a\in \Omega^-}\left(
\sum_{\substack{
j=1:n\\
v_j\notin S\cup\{u\}
}}\lW_{ij}\hy_j(a)> \sum_{v_j\in S}\lW_{ij}- \lW_{ui}\right)
\geq 1- \exp\left(\frac{\val_i(S,u)^2}{\sum_{\substack{
j=1:n\\
v_j\notin S\cup\{u\}
}}\lW_{ij}^2}\right)~.
\]

Putting together, we obtain:
If $\val_i(S,u)< 0$:
\begin{align*}
    \Gain_i(S,u)=&
\Prob\left(
a\in \Omega^+
\right)
\Prob_{a\in \Omega^+}\left(
\sum_{\substack{
j=1:n\\
v_j\notin S\cup\{u\}
}}\lW_{ji}\hy_j(a)< -\sum_{v_j\in S}\lW_{ji}+ \lW_{ui}\right)\\
&\ \ \ +
\Prob\left(
a\in \Omega^-
\right)
\Prob_{a\in \Omega^-}\left(
\sum_{\substack{
j=1:n\\
v_j\notin S\cup\{u\}
}}\lW_{ji}\hy_j(a)> \sum_{v_j\in S}\lW_{ji}- \lW_{ui}
\right)~\\
\geq &   [1-\exp\left(-\frac{\val_i(S,u)^2}{\sum_{\substack{
j=1:n\\
v_j\notin S\cup\{u\}
}}\lW_{ji}^2}\right)][\Prob(a\in \Omega^+)+\Prob(a\in \Omega^-)]=1-\exp\left(-\frac{\val_i(S,u)^2}{\sum_{\substack{
j=1:n\\
v_j\notin S\cup\{u\}
}}\lW_{ji}^2}\right) 
\end{align*}

Note that $\Gamma_i(S,u)\leq 1$. Therefore, if $\val_i(S,u)< 0$, we define $\widehat{\Gamma}_i(S,u)=1$ and we will have:

\[0\leq \widehat{\Gamma}_i(S,u)- \Gamma_i(S,u)\leq  \exp\left(-\frac{\val_i(S,u)^2}{\sum_{\substack{
j=1:n\\
v_j\notin S\cup\{u\}
}}\lW_{ji}^2}\right) \leq \exp\left(-\frac{\val_i(u)^2}{\sum_{j=1}^n\lW_{ji}^2}\right)  ~.\]

This completes our proof. 
\end{proof}

\subsubsection{\textbf{Missing material from \cref{sec:algo:ind} related to ambiguous vertices}}
\begin{lemma}\label{lem:usedotprod}
Let $V^+$ be those agents with error less than $1/2$ and $V^-$ be those agents with error greater than $1/2$. i.e,
\[V^+=\{v_j\mid \err(v_j)\leq 1/2\} \quad \& \ V^-=\{v_j\mid \err(v_j)> 1/2\}\]

and consider the following vectors 
\[ \lW^+_i=(\lW_{ij})_{v_j\in V^+} \quad \& ~  {\mathcal E}^+= \left(1-2\err(v_j)\right)_{v_j\in V^+} \]
and 
\[ \lW^-_i=(\lW_{ij})_{v_j\in V^-} \quad \& ~  {\mathcal E}^-= \left(2\err(v_j)-1\right)_{v_j\in V^-} \]
and 
\[\lW_i=(\lW_{i1},\lW_{i2},\dots, \lW_{in} )\]

We have that: 
 \begin{align*}
 \exp\left(
    -\frac{\val_i(u)^2}{4\sum_{i=1}^n\lW_{ij}^2}
    \right)&\leq \exp \left(-\frac{1}{4}\left(\frac{
 \langle \lW_i^+, {\mathcal E}^+\rangle 
 }{\abs{\lW_{i}}_2} 
 - \frac{
 \langle \lW_i^-, {\mathcal E}^-\rangle 
 }{\abs{\lW_{i}}_2} - 2
 \right)^2\right)\\
    &
    \leq 
    \exp\left(- \frac{1}{4}\left(
 M \cdot \frac{\abs{\lW_i^+}_2}{\abs{\lW_i}_2}\cdot \abs{{\mathcal E}^+}_2-\frac{\abs{\lW_i^-}_2}{\abs{\lW_i}_2}\cdot \abs{{\mathcal E}^-}_2-2 
 \right)^2\right)
 \end{align*}
    where 
    \[
    M=\frac{\max{\lW_i^+}}{\min{\lW_i^+}}\cdot \frac{1}{\min{\mathcal E}_i^+}~.
    \]
\end{lemma}
\begin{proof}
We show the above equation by finding an lower bound for $\frac{\val_i(u)^2}{\sum_{i=1}^n\lW_{ij}^2}$.

Consider an arbitrary vector $X$the $\ell_2$ norm is defined as:
\[
\abs{X}_2= \sqrt{\sum_{x_j\in X}x_j^2}
\]

Note that all of the above  vectors only have positive elements. 

We have:
\begin{align}
 \frac{\val_i(u)^2}{\sum_{i=1}^n\lW_{ij}^2}&=  \left(\frac{\val_i(u)}{\abs{\lW_{i}(V)}_2} \right)^2\nonumber \\
 &= 
 \left(\frac{
 \sum_{v_i\in V^+} \lW_{ij}(1-2\err(v_j))
 }{\abs{\lW_{i}(V)}_2} 
 - \frac{
 \sum_{v_i\in V^-} \lW_{ij}(2\err(v_j)-1)
 }{\abs{\lW_{i}(V)}_2} - \frac{2\err(u)\lW_{iu}}{\abs{\lW_{i}(V)}_2}
 \right)^2\nonumber \\
 &= 
 \left(\frac{
 \langle \lW_i^+, {\mathcal E}^+\rangle 
 }{\abs{\lW_{i}}_2} 
 - \frac{
 \langle \lW_i^-, {\mathcal E}^-\rangle 
 }{\abs{\lW_{i}}_2} - \frac{2\err(u)\lW_{iu}}{\abs{\lW_{i}}_2}
 \right)^2
\end{align}

where in the last equation $\langle \rangle$ denotes dot product.

Using P$\rm \acute{o}$lya-Szeg$\rm \ddot{o}$’s inequality we have: 
\[
\frac{
 \langle \lW_i^+, {\mathcal E}^+\rangle 
 }{\abs{\lW_{i}}_2} \geq \frac{\abs{\lW_i^+}_2}{\abs{\lW_i}_2}\cdot \frac{\max{\lW_i^+}}{\min{\lW_i^+}}\cdot \frac{\abs{{\mathcal E}^+}_2}{\min{\mathcal E}^+}
\]

Using Cauchy Schwarz we have 
\[
\frac{
 \langle \lW_i^-, {\mathcal E}^-\rangle 
 }{\abs{\lW_{i}}_2} \leq \frac{\abs{\lW_i}_2}{\abs{\lW_i^-}_2}\cdot {\abs{{\mathcal E}^-}_2}
\]

Therefore, letting $M=\frac{\max{\lW_i^+}}{\min{\lW_i^+}}\cdot \frac{1}{\min{\mathcal E}_i^+}$

we have: 
\begin{align*}
 \frac{\val_i(u)^2}{\sum_{i=1}^n\lW_{ij}^2}&\geq   \left(
 M \cdot \frac{\abs{\lW_i^+}_2}{\abs{\lW_i}_2}\cdot \abs{{\mathcal E}^+}_2-\frac{\abs{\lW_i^-}_2}{\abs{\lW_i}_2}\cdot \abs{{\mathcal E}^-}_2  - \frac{2\err(u)\lW_{iu}}{\abs{\lW_{i}}_2}
 \right)^2
\end{align*}

The premise may be concluded from the fact that $\frac{\err(u)\lW_{iu}}{\abs{\lW_{i}}_2}\leq 1$.

\end{proof}

\begin{proof}[\textbf{Proof of \cref{lemm:egalapp:good}}]
From \cref{lem:appgain} we now that 

    \[
    \abs{\widehat{\Gain_i}(S,u)-\Gain_i(S,u)}\leq  \exp\left(
    -\frac{\val_i(u)^2}{4\sum_{i=1}^n\lW_{ij}^2}
    \right)~,
       \]

       and from \cref{lem:usedotprod} we have: 

        \[\exp\left(
    -\frac{\val_i(u)^2}{4\sum_{i=1}^n\lW_{ij}^2}
    \right)\leq 
    \exp\left(- \frac{1}{4}\left(
 M \cdot \frac{\abs{\lW_i^+}_2}{\abs{\lW_i}_2}\cdot \abs{{\mathcal E}^+}_2-\frac{\abs{\lW_i^-}_2}{\abs{\lW_i}_2}\cdot \abs{{\mathcal E}^-}_2-2 
 \right)^2\right)
    \]

    Putting together and assume that $v_i$ is not ambiguous. We have that: 
\begin{align*}
       \abs{\widehat{\Gain_i}(S,u)-\Gain_i(S,u)}&\leq \exp\left(
    -\frac{\val_i(u)^2}{4\sum_{i=1}^n\lW_{ij}^2}
    \right)\leq 
    \exp\left(- \frac{1}{4}\left(
 M \cdot \frac{\abs{\lW_i^+}_2}{\abs{\lW_i}_2}\cdot \abs{{\mathcal E}^+}_2-\frac{\abs{\lW_i^-}_2}{\abs{\lW_i}_2}\cdot \abs{{\mathcal E}^-}_2-2 
 \right)^2\right)\\
 &\leq \exp\left(
 -\frac{1}{4}\left(3\sqrt{\log n} -2 \right)^2
 \right)\leq \exp(-\frac{1}{4}(5 \log n ))= o(n^{-1})~.
\end{align*}
\end{proof}

\subsubsection{\textbf{Missing material from \cref{sec:algo:ind}: proof of the main theorems}}\label{app:sec:proofs:egal}

In this subsection we present a pseudocode of our algorithm under the assumption that the agents are pairwise independent.

The following theorem bounds the error of this algorithm:

\begin{theorem}\label{thm:corr:firstbound}
    Assume that $S$ is the output of \cref{alg:2c}. We have that:
    \[
\Gegal(S)\geq [(1-1/e)-\Delta^{\rm ind}]\cdot \OPT{egal}~,
\]
where 
\[\Delta^{\rm ind}=\sum_{i=1}^n \exp\left(-\frac{\val_i(u)^2
}{\sum_{i=1}^n\lW_{ij}^2}\right)~.\]
\end{theorem}
\begin{proof}
The proof immediately follows from the fact that we have a submodular and monotone function and that the error greedy choice taken as 
\[g=\underset{u}{\rm argmax}\sum_{i=1}^n \widehat{\Gain}_i(S,u)~, \]
and that for any $v_i$ we have 
   \[\abs{\widehat{\Gain_i}(S,u)-\Gain_i(S,u)}\leq  \exp\left(
    -\frac{\val_i(u)^2}{4\sum_{i=1}^n\lW_{ij}^2}
    \right)~,
       \]
\end{proof}

\begin{proof}[\textbf{Proof of \cref{thm:egalcorerct} and \cref{thm:indep} }]
Using the above theorem \cref{thm:egalcorerct} and \cref{thm:indep} can be directly concluded just by noticing that for any non-ambiguous vertex the error induced on the greedy choice is at most $o(n^{-1})$. Thus, in total all of non-ambiguous vertices together will have an error of $o(1)$. The ambiguous each can have an error as large as $1$.
\end{proof}

\paragraph{Assuming having access to approximation $\widehat{\lW}$.}

If we only have access to $\widehat{\lW}$, we may run $\mathsf{EgalAlgo}^{\rm ind}$ using  $\widehat{\lW}$. In this case the approximation guarantee can be concluded from the following lemma which is similar to \cref{approxGamma}:

\begin{lemma}\label{lem:approxW:ind}
    Assume that all the all the classifiers 
    $\hy_1(a), \hy_2(a),\dots, \hy_n(a) $ 
    are pairwise  independent. We can estimate $\Gamma_i(S,u)$ (See \cref{gammaeq}) as follows:
    $$
    \widehat{\Gamma}_i(S,u)=
    \begin{cases}
    0& \text{ if } \widehat{\val}_i(S,u)>0\\
    1& \text{ otherwise}
    \end{cases}
    $$
    where, 
    \[
    \widehat{\val}_i(S,u)=\sum_{v_j\in S}\widehat{\lW}_{ij}+\sum_{\substack{
    j=1:n\\
    j\not \in S\cup\{u\} 
    }}\widehat{\lW}_{ij}[1-2\err(j)]-\widehat{\lW}_{iu}
    \]

    The error of this estimation is bounded by:

    \[
    \abs{\widehat{\Gamma}_i(S,u)-\Gamma_i(S,u)}\leq  \exp\left(
    -\frac{\val_i(u)^2}{4\sum_{i=1}^n\lW_{ij}^2}
    \right)(1+\epsilon)~,
    \]    
    where 
    \[
    \val_i(u)= \sum_{j=1}^n \lW_{ij}[1-2\err(j)]-2\err(u)\lW_{iu}~.
    \]
\end{lemma}
\begin{proof}
Similar to the proof of \cref{approxGamma} we may bound the two terms of $\Gamma_i(S,u)$ as follows:
\begin{align*}
& \Prob_{a\in \Omega^+}\left(
\sum_{\substack{
j=1:n\\
v_j\notin S\cup\{u\}
}}{\lW}_{ij}\hy_j(a)< -\sum_{v_j\in S}\lW_{ij}+ \lW_{ui}\right)\\
\leq &  \Prob_{a\in \Omega^+}\left(
\sum_{\substack{
j=1:n\\
v_j\notin S\cup\{u\}
}}\lW_{ij}\hy_j(a)<
\Ex\left[\sum_{\substack{
j=1:n\\
v_j\notin S\cup\{u\}
}}{\lW}_{ij}\hy_j(a)\right]- 
\widehat{\val}_i(S,u)-2\epsilon 
\right) 
\end{align*}

Similarly to the previous case if $\widehat{\val}_i(S,u)>0$ we may use Hoeffding bound to bound the above probability as:
\[
\exp\left(
-\frac{(\widehat{\val}_i(u)-2\epsilon)^2}{\sum_{j=1}^n \lW_{ij}^2}
\right)\leq \exp\left(
-\frac{(\val_i(u)-4\epsilon)^2}{\sum_{j=1}^n \lW^2_{ij}}
\right)
\]

Let's now bound RHS:
\begin{align*}
   \exp\left(
-\frac{(\val_i(u)-4\epsilon)^2}{\sum_{j=1}^n \lW_{ij}^2}
\right) &\leq   \exp\left(
-\frac{\val_i(u)^2}{\sum_{j=1}^n \lW^2_{ij}}
+ 8\epsilon \frac{\val_i(u)}{\sum_{j=1}^n \lW^2_{ij}}
\right)\leq  \exp\left(
-\frac{\val_i(u)^2}{\sum_{j=1}^n \lW^2_{ij}}
+ 8\epsilon 
\right)\\
&\leq  \exp\left(
-\frac{\val_i(u)^2}{\sum_{j=1}^n \lW^2_{ij}} 
\right)\exp( 8\epsilon) \leq (1+8\epsilon )\exp\left(
-\frac{\val_i(u)^2}{\sum_{j=1}^n \lW^2_{ij}} 
\right)
\end{align*}

If $\widehat{\val}_i(S,u)<0$ we may write: 
\begin{align*}
&\Prob_{a\in \Omega^+}\left(
\sum_{\substack{
j=1:n\\
v_j\notin S\cup\{u\}
}}{\lW}_{ij}\hy_j(a)> -\sum_{v_j\in S}\lW_{ij}+ \lW_{ui}\right)\\
\leq & \Prob_{a\in \Omega^+}\left(
\sum_{\substack{
j=1:n\\
v_j\notin S\cup\{u\}
}}\lW_{ij}\hy_j(a)>
\Ex\left[\sum_{\substack{
j=1:n\\
v_j\notin S\cup\{u\}
}}{\lW}_{ij}\hy_j(a)\right]- 
\widehat{\val}_i(S,u)+2\epsilon 
\right) 
\end{align*}
and with a similar argument we obtain the same bound. Bounding the probability in the case where $a\in \Omega^-$ can be done similarly. And the rest of the proof holds similarly to proof of \cref{approxGamma}.
\end{proof}

\begin{proof}[\textbf{Proof of \cref{remark:appW} independent case} ] The remark is a direct conclusion of the fact that  the error of the greedy choice is bounded by 
\[
\sum_{i=1}^n \abs{\widehat{\Gamma}_i(S,u)- \Gamma_i(S,u)}\leq (1+\epsilon)\sum_{i=1}^n 
\exp\left(- \frac{\val_i(u)^2}{4\sum_{j=1}^n\lW_{ij}^2}
\right)~.
\]
\end{proof}

\subsection{Missing proofs from \cref{sec:algo:group}: Estimating greedy choice assuming group dependence}\label{app:sec:algo:group}
Assume the assumption presented in \cref{def:model} and  remember the following definition from previous sections 
\[
\Gain_i(S,u)=
\Prob_{a\sim \Omega}
\left(
\B(i,a)\leq 0  \wedge \left(\bigwedge_{
\substack{v_j\in S
\\ 
\lW_{ij}\neq  0 
} } \y(a)= \hy_j(a)\right)
\wedge \y(a)\neq \hy_u(a)
\right), 
\]
and our goal is to estimate $\Gain_i(S,u)$ for all $S \subseteq V$, $u\in V$ and $i\in [n]$.

We may write 

\[
\Gain_i(S,u)=\Gain^{\rm gr}_i(S,u)\rho+ \Gain^{\rm indv}_i(S,u)(1-\rho)~,\]

, where the super-scripts  show whether individual or  group decisions have been made. 
Estimation of $\Gain_i^{\rm indv}(S,u)$ can be obtained using \cref{lem:appgain} and by calling $\mathsf{EstGain}^{\rm ind}$ of \cref{alg:2c}.  Estimation of $\Gain^{\rm gr}_i(S,u)$ can be done through a series of lemmas that we present here, and it depends on the whether  vertices of $S_i$, defined as $S_i\triangleq \{v_j\in S\mid \lW_{ij}\neq 0\}$, intersects $\R$, $\Bl$ or both. A pseudocode is presented at \cref{alg:2c,proc:gr}. Our  analysis is presented in  \cref{sec:nocolot,sec:monochorS,sec:bichorS}. The following lemma summarizes these results:

\begin{lemma}\label{app:greedy:group}
Assume the model presented in \cref{def:model}.
Having a set $S$ let $g$ be the vertex which maximizes the following function 

\[
g= \underset{u\in V}{\rm argmax}\  \rho \sum_{\substack{ i=1:n\\ \W_{ij}\neq 0}}\widehat{\Gain}_i^{\rm gr}(S,u) + (1-\rho)\sum_{\substack{ i=1:n\\ \W_{ij}\neq 0}}\widehat{\Gain}_i^{\rm indv}(S,u)
\]

Where $\widehat{\Gain}_i^{\rm gr}(S,u)$ may be obtained from \cref{lem:noch,lem:monoch,lem:bich}, and $\widehat{\Gain}_i^{\rm indv}(S,u)$ may be obtained from \cref{lem:appgain}. We have that:

\begin{align*}
  \Gegal \left(
S\cup \{\Greedy(S)\}
\right)- \Gegal \left(
S\cup \{g\}
\right)  \leq 
\rho\sum_{i=1}^n \exp\left(
-\frac{(\val_i^\W(u)- \Delta\lW_i)^2}{
4\sum_{v_j\in \W} \lW_{ij}^2
}
\right)+(1-\rho) \sum_{i=1}^n \exp\left(
-\frac{(\val_i(u))^2}{
4\sum_{v_j\in V} \lW_{ij}^2
}
\right)
\end{align*}

\end{lemma}

We will use the following definitions and results throughout the section. 

\[
\Gain^{\rm gr}_i(S,u)=
\Prob^{\rm gr}_{a\sim \Omega}\left(
\B(i,a)\leq 0  \wedge \T_i(S)\wedge \F(u)
\right)~,
\]
where  $\T_i(S)$ and $\F(u)$ are the following events:
\[
\T_i(S)=\bigwedge_{
\substack{v_j\in S_i
} } \y(a)= \hy_j(a)\quad \& \ \F(u)=\y(a)\neq \hy_u(a)
~. 
\]

 for any $v_i$ and $X\subseteq V$ we denote: 
\[
\lW_i(X)=\sum_{v_j\in X}\lW_{ij}
\]

In addition, for any $X,Y \subseteq V$ we define,

\[
\val_i(X,Y)= \sum_{v_j\in V\setminus (X\cup Y)} [1-2\err^{\rm indv}(v_j)]+\lW_i(X)-\lW_i(Y)~.
\]

We will use the following lemma throughout:

\begin{lemma}\label{app:lem:est}
Let's define 
\[\Gamma_i^+(X,Y)=\Prob_{a\in \Omega^+}
\left(\sum_{v_j\in V\setminus (X\cup Y)} 
\lW_{ij}\ihy_j(a)\leq \lW_i(Y)-\lW_i(X)
\right)~,
\]
and 
\[\Gamma_i^-(X,Y)=\Prob_{a\in \Omega^-}
\left(\sum_{v_j\in V\setminus (X\cup Y)} 
\lW_{ij}\ihy_j(a)\geq \lW_i(X)-\lW_i(Y)
\right)~.
\]

We may estimate the above probabilities as 
$$
\widehat{\Gamma}_i^+(X,Y)=\widehat{\Gamma}_i^-(X,Y)= \begin{cases}
0 & \text{if }\  \val_i(X,Y)\geq 0\\
1 & \text{if }\  \val_i(X,Y)< 0
\end{cases}
$$

We have that:
\[
\abs{\Gamma_i^+(X,Y)-\widehat{\Gamma}_i^+(X,Y) }\leq \exp\left(-\frac{(\val_i(X,Y))^2}{\sum_{v_j\in V\setminus (X\cup Y)} 
\lW_{ij}^2 }\right)~,
\]
and 
\[
\abs{\Gamma_i^-(X,Y)-\widehat{\Gamma}_i^-(X,Y) }\leq \exp\left(-\frac{(\val_i(X,Y))^2}{\sum_{v_j\in V\setminus (X\cup Y)} 
\lW_{ij}^2 }\right)~.
\]
\end{lemma}
\begin{proof}
Assume $a\in \Omega^+$ in this case we have that 

\[\Ex_{a\in \Omega^+}\left[
\ihy_j(a)
\right]=1-2\err^{indv}(v_j)\]

Thus, 
\begin{align*}
    \sum_{v_j\in V\setminus (X\cup Y)} 
\lW_{ij}\ihy_j(a)\leq \lW_i(Y)-\lW_i(X)\iff \\
  \sum_{v_j\in V\setminus (X\cup Y)} 
\lW_{ij}\ihy_j(a)\leq \Ex\left[\sum_{v_j\in V\setminus (X\cup Y)} 
\lW_{ij}\ihy_j(a)\right]-\left(\Ex\left[\sum_{v_j\in V\setminus (X\cup Y)} 
\lW_{ij}\ihy_j(a)\right]-\lW_i(Y)+\lW_i(X)\right)\iff \\
  \sum_{v_j\in V\setminus (X\cup Y)} 
\lW_{ij}\ihy_j(a)\leq \Ex\left[\sum_{v_j\in V\setminus (X\cup Y)} 
\lW_{ij}\ihy_j(a)\right]-\val_i(X,Y)
\end{align*}

If $\val_i(X,Y)>0$ we may use the Hoeffding bound to obtain: 
\[
\Prob_{a\in \Omega^+}\left( \sum_{v_j\in V\setminus (X\cup Y)} 
\lW_{ij}\ihy_j(a)\leq \lW_i(Y)-\lW_i(X)\right)\leq \exp\left(
-\frac{\left(\val_i(X,Y)\right)^2}{ \sum_{v_j\in V\setminus (X\cup Y)} 
\lW_{ij}^2}
\right)~.
\]
If $- \val_i(X,Y)>0$,  we write: 
\begin{align*}
    \sum_{v_j\in V\setminus (X\cup Y)} 
\lW_{ij}\ihy_j(a)\geq \lW_i(Y)-\lW_i(X)\iff \\
  \sum_{v_j\in V\setminus (X\cup Y)} 
\lW_{ij}\ihy_j(a)\geq \Ex\left[\sum_{v_j\in V\setminus (X\cup Y)} 
\lW_{ij}\ihy_j(a)\right]-\val_i(X,Y)
\end{align*}
Thus, 
\[
\Prob_{a\in \Omega^+}
 \left( \sum_{v_j\in V\setminus (X\cup Y)} 
\lW_{ij}\ihy_j(a)\geq \lW_i(Y)-\lW_i(X)\right)
\leq 
\exp\left(
-\frac{\left(\val_i(X,Y)\right)^2}{ \sum_{v_j\in V\setminus (X\cup Y)} 
\lW_{ij}^2}
\right)
\]

Therefore, 
\begin{align*}
    &\Prob_{a\in \Omega^+}\left( \sum_{v_j\in V\setminus (X\cup Y)} 
\lW_{ij}\ihy_j(a)\leq \lW_i(Y)-\lW_i(X)\right)\\
=&
1- \Prob_{a\in \Omega^+}\left( \sum_{v_j\in V\setminus (X\cup Y)} 
\lW_{ij}\ihy_j(a)\geq \lW_i(Y)-\lW_i(X)\right)\\
\geq & 1-\exp\left(
-\frac{\left(\val_i(X,Y)\right)^2}{ \sum_{v_j\in V\setminus (X\cup Y)} 
\lW_{ij}^2}
\right)~.
\end{align*}

Putting together we have:
\[
\abs{\Gamma_i^+(X,Y)-\widehat{\Gamma}_i^+(X,Y) }\leq \exp\left(-\frac{(\val_i(X,Y))^2}{\sum_{v_j\in V\setminus (X\cup Y)} 
\lW_{ij}^2 }\right)~,
\]
Similarly if $a\in \Omega^-$ we have: 
\[\Ex_{a\in \Omega^-}\left[
\ihy_j(a)
\right]=2\err^{\rm indv}(v_j)-1\]
thus, 
\begin{align*}
    \sum_{v_j\in V\setminus (X\cup Y)} 
\lW_{ij}\ihy_j(a)\geq \lW_i(X)-\lW_i(Y)\iff \\
  \sum_{v_j\in V\setminus (X\cup Y)} 
\lW_{ij}\ihy_j(a)\geq \Ex\left[\sum_{v_j\in V\setminus (X\cup Y)} 
\lW_{ij}\ihy_j(a)\right]+\left(-\Ex\left[\sum_{v_j\in V\setminus (X\cup Y)} 
\lW_{ij}\ihy_j(a)\right]+\lW_i(X)-\lW_i(Y)\right)
\end{align*}
and 
\begin{align*}
-\Ex\left[\sum_{v_j\in V\setminus (X\cup Y)} 
\lW_{ij}\ihy_j(a)\right]+\lW_i(X)-\lW_i(Y)&=
    -\sum_{v_j\in V\setminus (X\cup Y)} 
\lW_{ij}[2\err^{\rm indv}(v_j)-1]+\lW_i(X)-\lW_i(Y)\\
&=\val_i(X,Y)~.
\end{align*}

The rest of the proof follows similarly to the previous case. 
\end{proof}

\subsubsection{Case 1. Colorless $S_i$}\label{sec:nocolot} 
Assume that $S_i\subseteq \W$,

\begin{lemma}\label{lem:noch} If $S_i\subseteq \W$ and $u\in \W$, for any arbitrary $v_i$ we may take:
\begin{align*}
    &\widehat{\Gain}^{\rm gr}_i(S,u)\\
    =&\begin{cases}
\err^{\rm indv}(u)\cdot \prod_{
\substack{v_j\in S
\\ 
\lW_{ij}\neq  0 
} } (1-\err^{\rm indv}(v_j))
& \quad  \text{if } \val_i(\R\cup S, \Bl\cup \{u\})<0\  \& \val_i(\Bl\cup S, \R\cup \{u\})<0\\
\err^{\rm indv}(u)\cdot \prod_{
\substack{v_j\in S
\\ 
\lW_{ij}\neq  0 
} } (1-\err^{\rm indv}(v_j))\cdot \err(\R)&  \quad \text{if } \val_i(\R\cup S, \Bl\cup \{u\})>0\  \& \val_i(\Bl\cup S, \R\cup \{u\})<0\\
\err^{\rm indv}(u)\cdot \prod_{
\substack{v_j\in S
\\ 
\lW_{ij}\neq  0 
} } (1-\err^{\rm indv}(v_j))\cdot \err(\Bl)
&  \quad \text{if } \val_i(\R\cup S, \Bl\cup \{u\})<0\  \& \val_i(\Bl\cup S, \R\cup \{u\})>0\\
0 &\quad \text{otherwise}
\end{cases}
\end{align*}
and we have:
\[
\abs{\widehat{\Gain}_i^{\rm gr}(S,u)-\Gain_i^{\rm gr} (S,u)}\leq \exp\left(
-\frac{(\val_i^\W(u)- \Delta\lW_i)^2}{
4\sum_{v_j\in \W} \lW_{ij}^2
}
\right)
\]

\end{lemma}
\begin{proof}
   \begin{align*}
       \Prob^{\rm gr}_{a\sim \Omega}\left(
       \B(i,a)\leq 0 \wedge \T_i(S)\wedge \F(u)
       \right) &=     \Prob^{\rm gr}_{a\sim \Omega^+}\left(
       \z(i,a)\leq 0 \wedge \T_i(S)\wedge \F(u)
       \right) \mathbb{P}(a\in \Omega^+) \\
       &\quad \quad +\Prob^{\rm gr}_{a\sim \Omega^-}\left(
       \z(i,a)\geq 0 \wedge \T_i(S)\wedge \F(u)
       \right) \mathbb{P}(a\in \Omega^-)~. 
   \end{align*} 
We have:
\begin{align*}
        \Prob^{\rm gr}_{a\sim \Omega^+}\left(
       \z(i,a)\leq 0 \wedge \T_i(S)\wedge \F(u)
       \right)&=    \Prob^{\rm gr}_{a\sim \Omega^+}\left(
       \z(i,a)\leq 0 \mid \T_i(S)\wedge \F(u)
       \right)\Prob_{a\sim \Omega^+}^{\rm gr}\left(\T_i(S)\wedge \F(u)\right)\ \& \\
        \Prob^{\rm gr}_{a\sim \Omega^-}\left(
       \z(i,a)\leq 0 \wedge \T_i(S)\wedge \F(u)
       \right)&=    \Prob^{\rm gr}_{a\sim \Omega^-}\left(
       \z(i,a)\leq 0 \mid \T_i(S)\wedge \F(u)
       \right)\Prob_{a\sim \Omega^-}^{\rm gr}\left(\T_i(S)\wedge \F(u)\right)
\end{align*}
Note that since $S\subseteq \W$ and $u\in \W$ we have:
\[
\Prob_{a\sim \Omega^-}^{\rm gr}\left(\T_i(S)\wedge \F(u)\right)=\Prob_{a\sim \Omega^+}^{\rm gr}\left(\T_i(S)\wedge \F(u)\right)= \err^{\rm indv}(u)\cdot \prod_{
\substack{v_j\in S
\\ 
\lW_{ij}\neq  0 
} } (1-\err^{\rm indv}(v_j))~.
\]
Furthermore, 
\begin{align*}
&\Prob^{\rm gr}_{a\sim \Omega^+}\left(
       \z(i,a)\leq 0 \mid \T_i(S)\wedge \F(u)
       \right)\\
       =& \Prob_{a\sim \Omega^+}^{\rm gr}\left(\sum_{\substack{j=1:n\\
       v_j\notin S\cup \{u\}} 
       }
       \lW_{ij}\hy_j(a)+ \sum_{v_j\in S} \lW_{ij} - \lW_{iu}\leq 0\right) \\
       =& \Prob_{a\sim \Omega^+}\left(\sum_{\substack{v_j\in \W\\
       v_j\notin S\cup \{u\}} 
       }
       \lW_{ij}\ihy_j(a)+ 
       \sum_{v_j\in \R
       }
       \lW_{ij}\ghy_j(a)+ \sum_{v_j\in \Bl
       }
       \lW_{ij}\ghy_j(a)+ \sum_{v_j\in S} \lW_{ij} - \lW_{iu}\leq 0\right) \\
       =& \Prob_{a\sim \Omega^+}\left(\sum_{\substack{v_j\in \W\\
       v_j\notin S\cup \{u\}} 
       }
       \lW_{ij}\ihy_j(a)+ 
       \sum_{v_j\in \R
       }
       \lW_{ij}- \sum_{v_j\in \Bl
       }
       \lW_{ij}+ \sum_{v_j\in S} \lW_{ij} - \lW_{iu}\leq 0\right)\err(\Bl) \\
       &
         \quad + \Prob_{a\sim \Omega^+}\left(\sum_{\substack{v_j\in \W\\
       v_j\notin S\cup \{u\}} 
       }
       \lW_{ij}\ihy_j(a)+ 
       \sum_{v_j\in \Bl
       }
       \lW_{ij}- \sum_{v_j\in \R
       }
       \lW_{ij}+ \sum_{v_j\in S} \lW_{ij} - \lW_{iu}\leq 0\right)\err(\R) \\
       =& \Gamma_i^+(\R\cup S, \Bl\cup\{u\})\err(\Bl)+ \Gamma_i^+(\Bl\cup S, \R\cup\{u\})\err(\R)~.
\end{align*}
Similarly we have that:
\begin{align*}
&\Prob^{\rm gr}_{a\sim \Omega^-}\left(
       \z(i,a)\geq 0 \mid \T_i(S)\wedge \F(u)
       \right)\\
       =& \Prob_{a\sim \Omega^-}^{\rm gr}\left(\sum_{\substack{j=1:n\\
       v_j\notin S\cup \{u\}} 
       }
       \lW_{ij}\hy_j(a)- \sum_{v_j\in S} \lW_{ij} + \lW_{iu}\geq 0\right) \\
       =& \Prob_{a\sim \Omega^-}\left(\sum_{\substack{v_j\in \W\\
       v_j\notin S\cup \{u\}} 
       }
       \lW_{ij}\ihy_j(a)+ 
       \sum_{v_j\in \R
       }
       \lW_{ij}\ghy_j(a)+ \sum_{v_j\in \Bl
       }
       \lW_{ij}\ghy_j(a)- \sum_{v_j\in S} \lW_{ij} + \lW_{iu}\geq 0\right) \\
        =& \Prob_{a\sim \Omega^-}\left(\sum_{\substack{v_j\in \W\\
       v_j\notin S\cup \{u\}} 
       }
       \lW_{ij}\ihy_j(a)- 
       \sum_{v_j\in \R
       }
       \lW_{ij}+ \sum_{v_j\in \Bl
       }
       \lW_{ij}- \sum_{v_j\in S} \lW_{ij} + \lW_{iu}\geq 0\right)\err(\Bl) \\
       &
         \quad + \Prob_{a\sim \Omega^-}\left(\sum_{\substack{v_j\in \W\\
       v_j\notin S\cup \{u\}} 
       }
       \lW_{ij}\ihy_j(a)- 
       \sum_{v_j\in \Bl
       }
       \lW_{ij}+ \sum_{v_j\in \R
       }
       \lW_{ij}- \sum_{v_j\in S} \lW_{ij} + \lW_{iu}\geq 0\right)\err(\R) \\
       =& \Gamma_i^-(\R\cup S,\Bl\cup\{u\})\err(\Bl)+ \Gamma_i^-(\Bl\cup S,\R\cup\{u\})\err(\R)~.
\end{align*}

Therefore, $\Prob^{\rm gr}_{a\sim \Omega}\left(
       \B(i,a)\leq 0 \wedge \T_i(S)\wedge \F(u)
       \right)$ is equal to

    \begin{align*}
         \err^{\rm indv}(u)\cdot \prod_{
\substack{v_j\in S
\\ 
\lW_{ij}\neq  0 
} } (1-\err^{\rm indv}(v_j))\cdot (
    [\Gamma_i^+(\R\cup S,\Bl\cup\{u\})\err(\Bl)+ \Gamma_i^+(\Bl\cup S,\R\cup\{u\})\err(\R)]\Prob(a\in \Omega^+)\\
    \quad +\  [
    \Gamma_i^-(\R\cup S,\Bl\cup\{u\})\err(\Bl)+ \Gamma_i^-(\Bl\cup S,\R\cup\{u\})\err(\R)]\Prob(a\in \Omega^-)
    )
    \end{align*}

We may now employ \cref{app:lem:est} to estimate the above probabilities. 

By replacing the estimations $\widehat{\Gamma}_i^+$ and $\widehat{\Gamma}_i^-$ in the above formula we obtain:
\begin{align*}
& \widehat{\Gain}^{\rm gr}_i(S,u)
\\
=& \begin{cases}
\err^{\rm indv}(u)\cdot \prod_{
\substack{v_j\in S
\\ 
\lW_{ij}\neq  0 
} } (1-\err^{\rm indv}(v_j))
& \quad  \text{if } \val_i(\R\cup S, \Bl\cup \{u\})<0\  \& \val_i(\Bl\cup S, \R\cup \{u\})<0\\
\err^{\rm indv}(u)\cdot \prod_{
\substack{v_j\in S
\\ 
\lW_{ij}\neq  0 
} } (1-\err^{\rm indv}(v_j))\cdot \err(\R)&  \quad \text{if } \val_i(\R\cup S, \Bl\cup \{u\})>0\  \& \val_i(\Bl\cup S, \R\cup \{u\})<0\\
\err^{\rm indv}(u)\cdot \prod_{
\substack{v_j\in S
\\ 
\lW_{ij}\neq  0 
} } (1-\err^{\rm indv}(v_j))\cdot \err(\Bl)
&  \quad \text{if } \val_i(\R\cup S, \Bl\cup \{u\})<0\  \& \val_i(\Bl\cup S, \R\cup \{u\})>0\\
0 &\quad \text{otherwise}
\end{cases}
\end{align*}
\end{proof}

\begin{lemma} If $S_i\subseteq \W$ and $u\in \R$, for any arbitrary $v_i$ we may take:

$$
\widehat{\Gain}^{\rm gr}_i(S,u)=\begin{cases}
\err^{\rm indv}(u)\cdot \prod_{
\substack{v_j\in S
\\ 
\lW_{ij}\neq  0 
} } (1-\err^{\rm indv}(v_j))
& \quad  \text{if } \val_i(\Bl\cup S, \R\cup \{u\})<0 \\
0 &\quad \text{otherwise}
\end{cases}
$$

If $u\in \Bl$, for any arbitrary $v_i$ we may take:

$$
\widehat{\Gain}^{\rm gr}_i(S,u)=\begin{cases}
\err^{\rm indv}(u)\cdot \prod_{
\substack{v_j\in S
\\ 
\lW_{ij}\neq  0 
} } (1-\err^{\rm indv}(v_j))
& \quad  \text{if } \val_i(\R\cup S, \Bl\cup \{u\})<0 \\
0 &\quad \text{otherwise}
\end{cases}
$$

and in both cases we have:
\[
\abs{\widehat{\Gain}_i^{\rm gr}(S,u)-\Gain_i^{\rm gr} (S,u)}\leq \exp\left(
-\frac{(\val_i^\W(u)- \Delta\lW_i)^2}{
4\sum_{v_j\in \W} \lW_{ij}^2
}
\right)~.
\]

\end{lemma}
\begin{proof}
Like previous lemma we have:
\begin{align*}
&\Prob^{\rm gr}_{a\sim \Omega^+}\left(
       \z(i,a)\leq 0 \mid \T_i(S)\wedge \F(u)
       \right)\\
       =& \Prob_{a\sim \Omega^+}\left(\sum_{\substack{v_j\in \W\\
       v_j\notin S\cup \{u\}} 
       }
       \lW_{ij}\ihy_j(a)+ 
       \sum_{v_j\in \R
       }
       \lW_{ij}\ghy_j(a)+ \sum_{v_j\in \Bl
       }
       \lW_{ij}\ghy_j(a)+ \sum_{v_j\in S} \lW_{ij} - \lW_{iu}\leq 0\right)
       \end{align*}
       If $u\in \R$ since we are conditioning on $\F(u)$ the above probability is equal to:
       \begin{align*}
\Prob^{\rm gr}_{a\sim \Omega^+}\left(
       \z(i,a)\leq 0 \mid \T_i(S)\wedge \F(u)
       \right)
       &= \Gamma_i^+(\Bl\cup S, \R\cup\{u\})~.
       \end{align*}
       Similarly 
       \begin{align*}
\Prob^{\rm gr}_{a\sim \Omega^+}\left(
       \z(i,a)\geq 0 \mid \T_i(S)\wedge \F(u)
       \right)
       &= \Gamma_i^-(\Bl\cup S, \R\cup\{u\})~.       \end{align*}
          If $u\in \Bl$ we have 

                 \begin{align*}
\Prob^{\rm gr}_{a\sim \Omega^+}\left(
       \z(i,a)\leq 0 \mid \T_i(S)\wedge \F(u)
       \right)
       &= \Gamma_i^+(\R\cup S, \Bl\cup\{u\})~.
       \end{align*}
       Similarly 
       \begin{align*}
\Prob^{\rm gr}_{a\sim \Omega^+}\left(
       \z(i,a)\geq 0 \mid \T_i(S)\wedge \F(u)
       \right)
       &= \Gamma_i^-(\R\cup S, \Bl\cup\{u\})~.       \end{align*}

       Putting together, and employing \cref{app:lem:est} we get the premise.
       
\end{proof}

\subsubsection{Case 2. monochromatic $S_i$}\label{sec:monochorS}

\begin{lemma}\label{lem:monoch}
Let $\C$ be either $\R$ or $\Bl$ and $\bC$ be the other color. 
    Assume $S_i\cap \C \neq \emptyset$ and $S_i\cap \bC = \emptyset$. In this case, for any $u\in \C$, we have $\forall i, \Gain^{\rm gr}_i(S,u)=0$. 
    
    For $u\in \bC$ we may use the following estimation:

    \[
    \widehat{\Gain}_i^{\rm gr}\left( S,u\right)=\begin{cases}
    \err(\bC)\prod_{v_j\in S_i\cap \W}(1-\err^{\rm indv}(v_j)) & \text{ if }\quad \val_i(\C \cup S,\bC )<0\\
    0 &\text{ otherwise ~.}
    \end{cases}
    \]

    and for $u\in \W$ we may use the following estimation: 

        \[
    \widehat{\Gain}_i^{\rm gr}\left( S,u\right)=\begin{cases}
    \err(u)\err(\bC)\prod_{v_j\in S_i\cap \W}(1-\err^{\rm indv}(v_j)) & \text{ if }\quad \val_i(\C \cup S,\bC \cup\{u\})<0\\
    0 &\text{ otherwise ~.}
    \end{cases}
    \]

    The above estimations satisfy:

    \[
    \abs{\widehat{\Gain}^{\rm gr}_i(S,u)-\Gain^{\rm gr}_i(S,u)
    }
    \leq  \exp\left(
-\frac{(\val_i^\W(u)- \Delta\lW_i)^2}{
4\sum_{v_j\in \W} \lW_{ij}^2
}
\right)~.
    \]
    
\end{lemma}

\begin{proof}[Proof of  \cref{lem:monoch}]
Assume that $S_i$ intersects only with one color $\C$ and its intersection with the other color $\bC$ is empty. 

If $u\in \C$, then $\Prob(\T_i(S)\wedge \F(u))=0$. Thus, $\forall i\  \Gain^{\rm gr}_i(S,u)=0$.

If $u\in \bC$:

We have that $\Prob(\T_i(S)\wedge \F(u))=\err(\bC)\cdot \prod_{v_j\in S_i\cap \W}(1-\err^{\rm indv}(v_j))$ . Furthermore, 
\begin{align*}
    \Prob^{\rm gr}_{a\in \Omega^+}\left(
    \z^*(i,a)\leq 0 \mid \T_i(S)\wedge \F(u)
    \right)&=\Prob
    \left(
    \sum_{v_j\in  \W\setminus S_i} \lW_{ij}\ihy_j(a)+ \lW(\C)-\lW(\bC)+\lW(S\cap \W)\leq 0
    \right)\\
    &=\Gamma_i^+\left(
    \C\cup S ,\bC
    \right)~.
\end{align*}

Similarly, 
\begin{align*}
    \Prob^{\rm gr}_{a\in \Omega^-}\left(
    \z^*(i,a)\geq 0 \mid \T_i(S)\wedge \F(u)
    \right)=\Gamma_i^-\left(
    \C\cup S,\bC
    \right)~.
\end{align*}
Putting together and employing \cref{app:lem:est} we conclude the first part of the premise. 

If $u\in \W$: 
We have that $\Prob(\T_i(S)\wedge \F(u))=\err(u)\err(\bC)\cdot \prod_{v_j\in S_i\cap \W}(1-\err^{\rm indv}(v_j))$. Furthermore, 
\begin{align*}
    \Prob^{\rm gr}_{a\in \Omega^+}\left(
    \z^*(i,a)\leq 0 \mid \T_i(S)\wedge \F(u)
    \right)&=\Prob
    \left(
    \sum_{v_j\in  \W\setminus S_i} \lW_{ij}\ihy_j(a)+ \lW_i(\C)-\lW_i(\bC)+\lW_i(S\cap \W)-\lW_{iu}\leq 0
    \right)\\
    &=\Gamma_i^+\left(
    \C\cup S,\bC\cup \{u\}
    \right)~.
\end{align*}

Similarly, 
\begin{align*}
    \Prob^{\rm gr}_{a\in \Omega^-}\left(
    \z^*(i,a)\geq 0 \mid \T_i(S)\wedge \F(u)
    \right)
    &=\Gamma_i^-\left(
    \C\cup S,\bC\cup \{u\}
    \right)~.
\end{align*}

Putting together and employing \cref{app:lem:est} we conclude the second part of the premise. 
\end{proof}

\subsubsection{Case 3. bichromatic $S_i$}\label{sec:bichorS}

\begin{lemma}\label{lem:bich}
    If $S_i\cap \R \neq \emptyset$ and $S_i\cap \Bl \neq  \emptyset$, we have $\forall i\  \Gain^{\rm gr}_i(S,u)=0$.
\end{lemma}
\begin{proof}
Note that $\Prob(\T_i(S))=0$. Thus, we conclude the premise. 
\end{proof}

\subsubsection{\textbf{Missing material from \cref{sec:algo:group}:\ $\boldsymbol{\W}$-Ambiguous vertices.}}\label{app:sec:Wambig}

Let's first present the definition of $\W$-ambiguous vertices in detail:

Consider the following partitioning of  white vertices to low and high error parts: 
\[\W^+=\{v_j\mid \err(v_j)\leq 1/2\} \ \ \& \ \W^-=\{v_j\mid \err(v_j)> 1/2\}\]
with respect to this partition we define the following vectors:
\[
 {\mathcal E^\W}^+{=}\left(1-2\err(v_j)\right)_{v_j\in \W^+}\ \ \& \ \quad   {\mathcal E^\W}^-{=} \left(2\err(v_j)-1\right)_{v_j\in \W^-}
\]
\[
\text{ and }\ \quad \quad \quad \quad   
{\lW^{\W+}_i}=(\lW_{ij})_{v_j\in \W^+} \ \   \&  \quad   \ {\lW^{\W-}_i}=(\lW_{ij})_{v_j\in \W^-}
\]

\begin{definition}
    [ $\W$-Ambiguous vertices ]
Let $\lW^\W_i=(\lW_{ij})_{v_j\in \W}$, and $\abs{\cdot}_2$ be the $\ell_2$ norm and $\langle \cdot, \cdot \rangle$ be dot product.

    We call an agent $v_i\in V$, \emph{$\W$-ambiguous}  if it satisfies 
 \[\abs{
 \frac{
 \langle{\lW^{\W+}_i},  {\mathcal E^\W}^+ \rangle 
 }{\abs{\lW_i^\W}_2}-
  \frac{
 \langle{\lW^{\W-}_i},  {\mathcal E^\W}^- \rangle 
 }{\abs{\lW_i^\W}_2}
 }\leq 4\sqrt{\log n}+ \Delta\lW_i~,
    \]
    where
$\Delta\lW_i=\abs{\sum_{v_j\in \R}\lW_{ij} - \sum_{v_j\in \Bl}\lW_{ij}}~.
    $
    A \emph{network} is \emph{nicely colored} if no vertex is $\W$-ambiguous.
\end{definition}

We now show that if a vertex is \emph{not} ambiguous w.r.t. $\W$ the group gain associated to it can be estimated with high precision:

\begin{lemma}
Let  $\widehat{\Gain_i}^{\rm gr}(S,u)$  be as defined in \cref{lem:noch,lem:monoch,lem:bich}.
If a vertex is \emph{non-ambiguous} w.r.t $\W$ we have: 
\[
\abs{\Gain^{\rm gr}_i(S,u)-\widehat{\Gain}^{\rm gr}_i(S,u)}\leq o(n^{-1})
~.\]
    
\end{lemma}
\begin{proof}
Note that from \cref{lem:noch,lem:monoch,lem:bich} we have that for any $v_i$: 
\[
\abs{\Gain^{\rm gr}_i(S,u)-\widehat{\Gain}^{\rm gr}_i(S,u)}\leq \exp\left(- \frac{(\val_i^\W(u)-\Delta\lW_i)^2}{4\sum_{v_j\in\W}\lW_{ij}^2}\right)
~.\]

Note that we have:

\begin{align*}
 \frac{(\val_i^{\W}(u)-\Delta\lW_i)^2}{\sum_{v_j\in\W}\lW_{ij}^2}&=  \left(\frac{\val_i^{\W}(u)-\Delta\lW_i}{\abs{\lW_{i}^\W}_2} \right)^2\\
 &= 
 \left(\frac{
 \langle \lW_i^{\W+}, {\mathcal E}^{\W+}\rangle 
 }{\abs{\lW_{i}^\W}_2} 
 - \frac{
 \langle \lW_i^{\W-}, {\mathcal E}^{\W-}\rangle 
 }{\abs{\lW_{i}^\W}_2} - \frac{2\err(u)\lW_{iu}}{\abs{\lW_{i}^\W}_2}- \frac{\Delta\lW_i}{\abs{\lW_{i}^\W}_2}
 \right)^2\\
 &\geq  \left(\frac{
 \langle \lW_i^{\W+}, {\mathcal E}^{\W+}\rangle 
 }{\abs{\lW_{i}^\W}_2} 
 - \frac{
 \langle \lW_i^{\W-}, {\mathcal E}^{\W-}\rangle 
 }{\abs{\lW_{i}^\W}_2} - \frac{\Delta\lW_i}{\abs{\lW_{i}^\W}_2}-2
 \right)^2
\end{align*}

Assuming that the vertex is not ambiguous w.r.t. whites we have:
\[\abs{
 \frac{
 \langle{\lW^{\W+}_i},  {\mathcal E^\W}^+ \rangle 
 }{\abs{\lW_i^\W}_2}-
  \frac{
 \langle{\lW^{\W-}_i},  {\mathcal E^\W}^- \rangle 
 }{\abs{\lW_i^\W}_2}
 }> 4\sqrt{\log n}+ \Delta\lW_i\]

 which implies 
 \[
  \frac{(\val_i^{\W}(u)-\Delta\lW_i)^2}{\sum_{v_j\in\W}\lW_{ij}^2}=  \left(\frac{\val_i^{\W}(u)-\Delta\lW_i}{\abs{\lW_{i}^\W}_2} \right)^2\geq (3\sqrt{\log n})^2\geq 5 \log n
 \]
 From which we conclude that the error is bounded as:
 \[
 \exp\left(-
   \frac{(\val_i^{\W}(u)-\Delta\lW_i)^2}{\sum_{v_j\in\W}\lW_{ij}^2}
 \right)\leq \exp(-5/4 \log n)= o(n^{-1})~.
 \]
\end{proof}

\subsubsection{\textbf{Missing material from \cref{sec:algo:group}: proof of the main theorems}}\label{app:thm:proofs:group}

\begin{proof}[\textbf{Proof of \cref{thm:egalgroupcorrectness} and \cref{remark:appW}}]
Note that the error of the greedy choice approximation is either generated from approximation of $\sum_{i=1}^n\widehat{\Gain}^{\rm gr}_i(S,u)$ or
$\sum_{i=1}^n\widehat{\Gain}^{\rm indv}_i(S,u)$.
The error of $\widehat{\Gain}^{\rm indv}_i(S,u)$ may be bounded similar to the independent case. To see the bound the error of $\sum_{i=1}^n\widehat{\Gain}^{\rm gr}_i(S,u)$ note that the error induced by all \emph{non-ambiguous} vertices is $o(n\times n^{-1})=o(1) $. The error of each ambiguous vertex is at most one, therefore, we conclude the result. 

\paragraph{Approximation when having access $\widehat{\W}$.} Follows similarly as in \cref{lem:approxW:ind}
\end{proof}

\subsection{Concentration Bounds}
\begin{theorem}\label{thm:hoeffding}
Let $X_1,X_2,\dots X_k$ be independent random variables in range $[-1,1]$ with means $\mu_1,\dots ,\mu_k$ and let $w_i$s be weight coefficients  . 
Taking  $\mu=\sum_{i=1}^k w_i \mu_i$. We have that for any $\varepsilon>0$:
 \[
 \Prob\left(\sum_{i=1}^k w_i X_i \geq \mu +\varepsilon \right)\leq \exp\left({-\frac{  \varepsilon^2}{4\sum_{i=1}^k w_i^2}} \right),~
 \]
 and 
 \[
 \Prob\left(\sum_{i=1}^k w_i X_i\leq \mu- \varepsilon \right)\leq \exp\left({-\frac{  \varepsilon^2}{4\sum_{i=1}^k w_i^2}} \right)~.
 \]
\end{theorem}

\section{Additional  experiments and details of set up}\label{asec:exp_details}
Our proposed problem and methods in general do not assume any prior knowledge of the given datasets. Our algorithms are deterministic with no parameter to tune. It has potential to be applied to any problems as long as the objective function of the problems of interests related to our proposed objective function.  In the following, we include the parameter settings of the problems in our experiments. These parameters are not tuned for our methods. 
We just fix the parameters to make sure the problems are nontrivial enough. We use the same parameters for all methods in our experiments.

We let $\abs{\Omega}=3$ and for initial opinion $\mathbf{\hy}$, we randomly generate for each agent $v_i$ a random vector $(\hy_i(a): a\in \Omega)$ with each entry sampled from Bernoulli distribution with probability $p_i$ of $\hy_i(a)=+1$ sampled uniformly from $[0.3, 0.9]$.
For random graph generators, 
some of the key parameters are fixed as follows:
\begin{itemize}
    \item Number of nodes: 128
    \item Erd\H{o}s-R\'{e}nyi graph: Probability for edge creation $p=0.005$.
    \item Barab\'{a}si-Albert preferential attachment model: Number of edges to attach from a new node to existing nodes $m=5$
    \item Watts-Strogatz small-world graph: Each node is joined with its $k=5$ nearest neighbors in a ring topology; The probability of rewiring each edge $p=0.25$
\end{itemize}


In practice, we found that the infinite FJ model with weight matrix converging to $(I+L)^{-1}$ makes the problem less interesting than general case in practice. The reason is that all entries of the matrix $(I+L)^{-1}$ are positive and non-zero. With such weight matrix, we found even a random selection algorithm can converge very fast. To avoid such simple cases, here we apply a finite $t$-step FJ model. We fix $t=3$ to ensure the induced matrix $\bar{W}$ is sparse enough.

Besides randomly generated graphs from tree classical models (ER, PA and WS), we also test our algorithms on random generated matrix $\lW$ directly, denoted as RandomW. Each entry of $\lW$ is independently sampled from a uniform distribution on $[0,1]$. We let the sparsity of $\lW$ to be around $0.95$. Each row of $\lW$ is normalized thus sum up to be $1$. 

The statistics of real and synthetic datasets are summarized in~\cref{tb:dataset_stats}. Experiment results are illustrated in~\cref{fig:exp_appendix}.

We let the set of faulty prediction be 
$\B^{f}\triangleq\mathbb{E}_{a\in \Omega}[\sum_{i=1}^n\mathbf{1}\B(i,a)<0]$~, which serves as an upperbound on the egalitarian improvement. 
 We define the accuracy {\cvratio } as: 
\[\cvratio=\frac{\Gegal(S)}{{\B^f}}~.\]

We calculate all expected values by taking averages over $a\in \Omega$. 

\begin{table}[ht]
\centering
\caption{Table of statistics of datasets. Sparsity of $\lW$ represents the percentage of zero entries in $\lW$. }
\label{tb:dataset_stats}
\begin{tabular}{@{}lccc@{}}
\toprule
        & size & sparsity of $\lW$ & faulty prediction ${\B^{f}}$ \\ \midrule
ER      & 128    & 0.98            & 74.67               \\
PA      & 128    & 0.37            & 32.0               \\
WS      & 128    & 0.74            & 38.67               \\
BIO     & 297     & 0.45            & 55.33                \\
CSPK    & 39      & 0.75            & 14.67                 \\
FB      & 620     & 0.75            & 161.33            \\
WIKI    & 890     & 0.66            & 238.67               \\
RandomW & 128     & 0.95            & 37.67                \\ \bottomrule
\end{tabular}
\end{table}

\begin{figure}[th]
  \centering
    \begin{tabular}[]{lcr}
    \includegraphics[width=.47\linewidth]{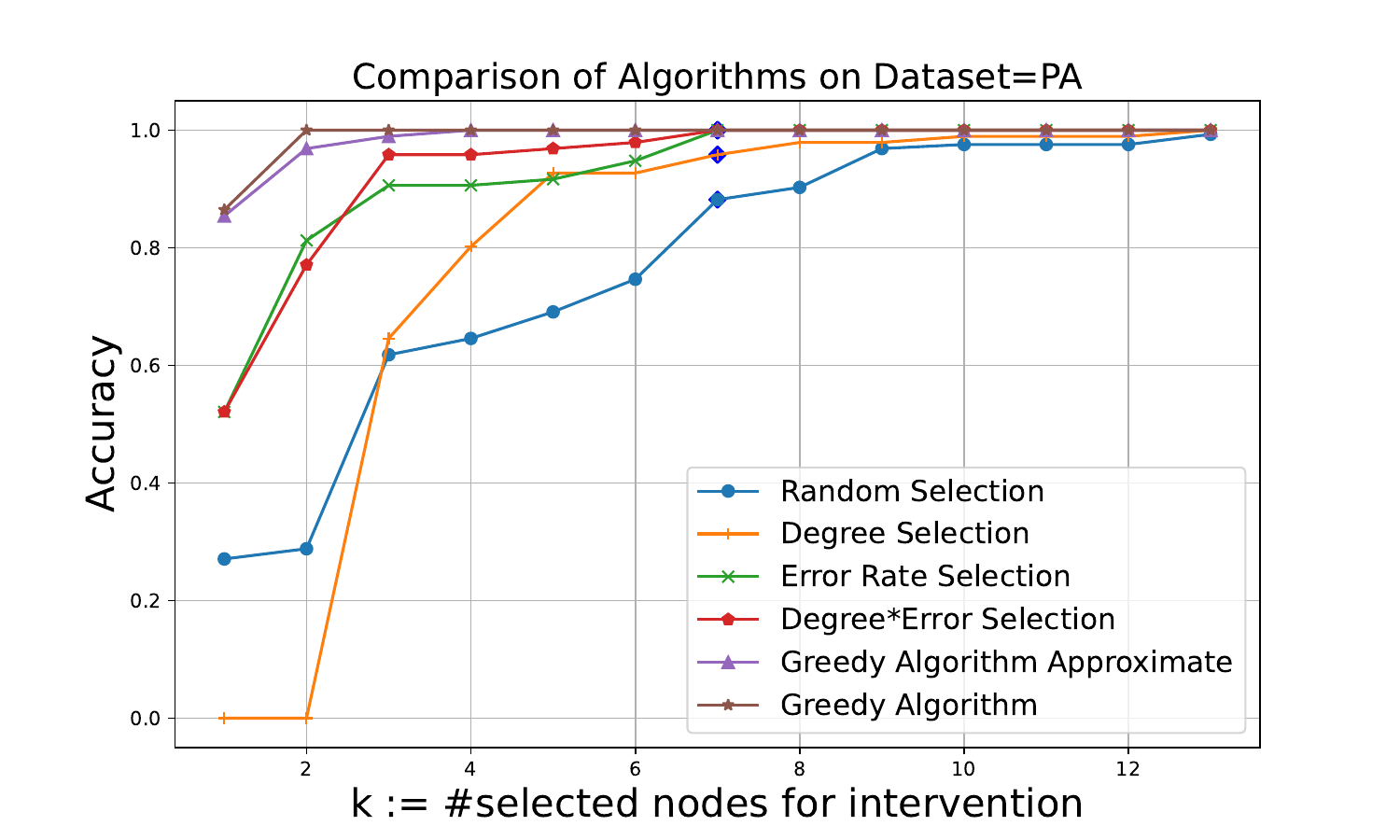}  & \includegraphics[width=.47\linewidth]{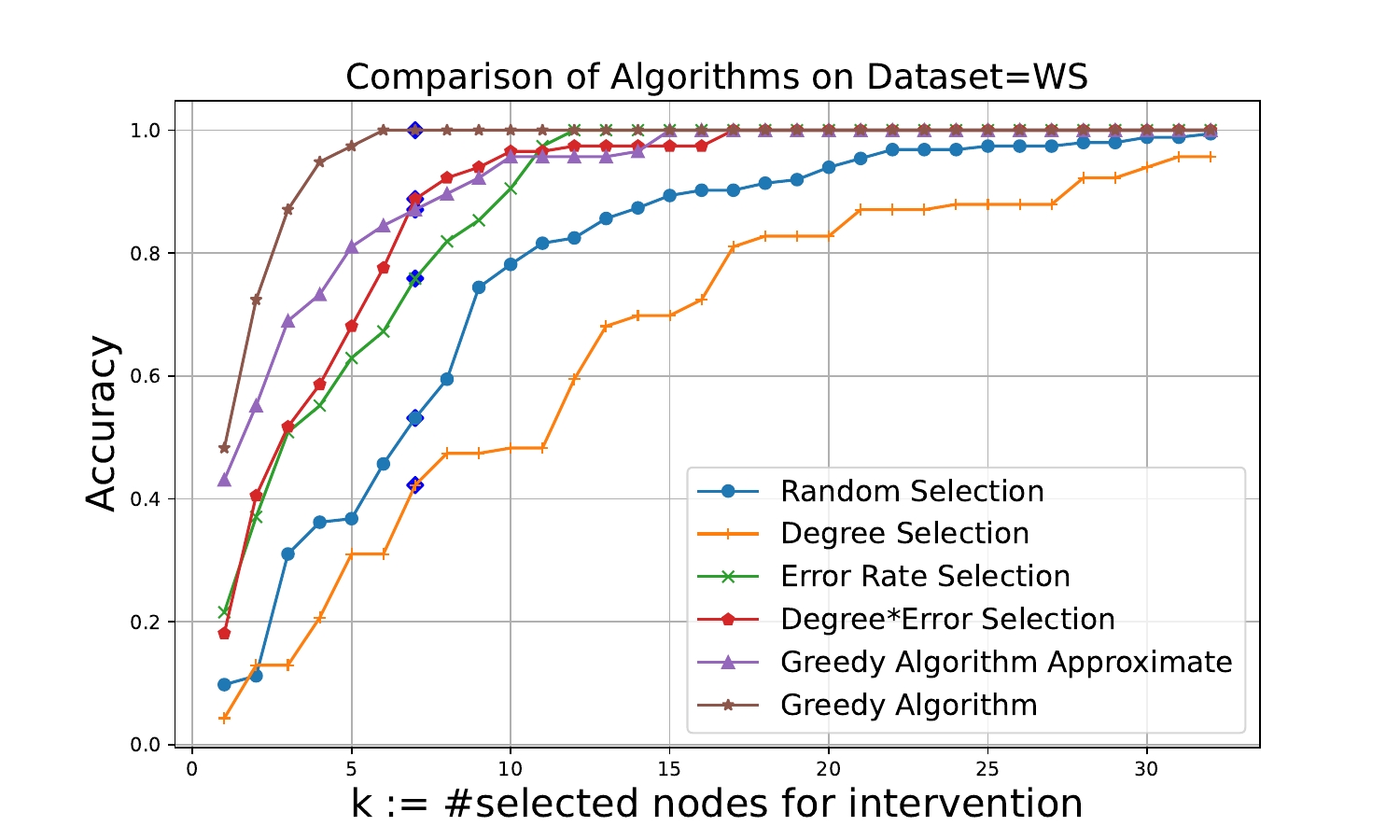} \\     \includegraphics[width=.47\linewidth]{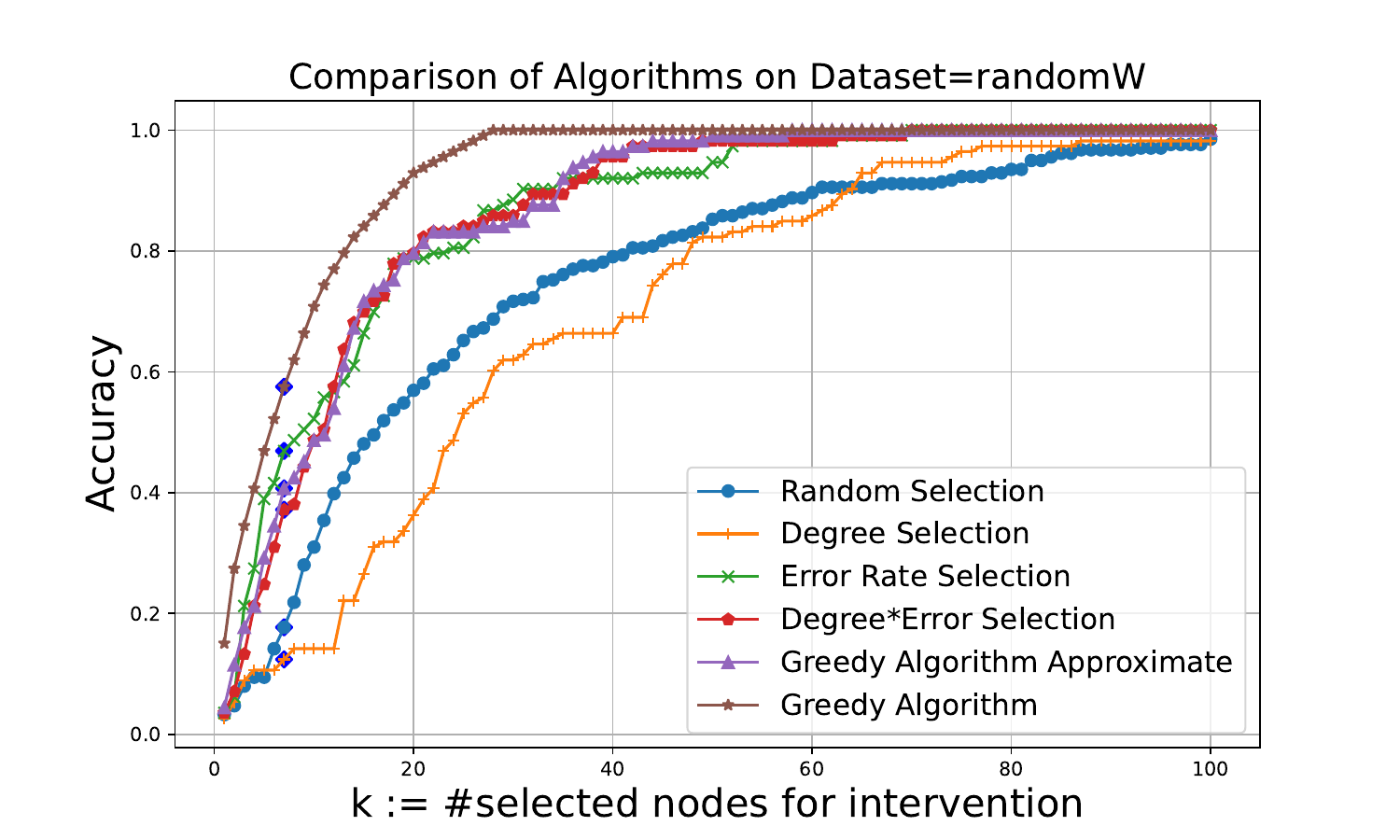} &
    \includegraphics[width=.47\linewidth]{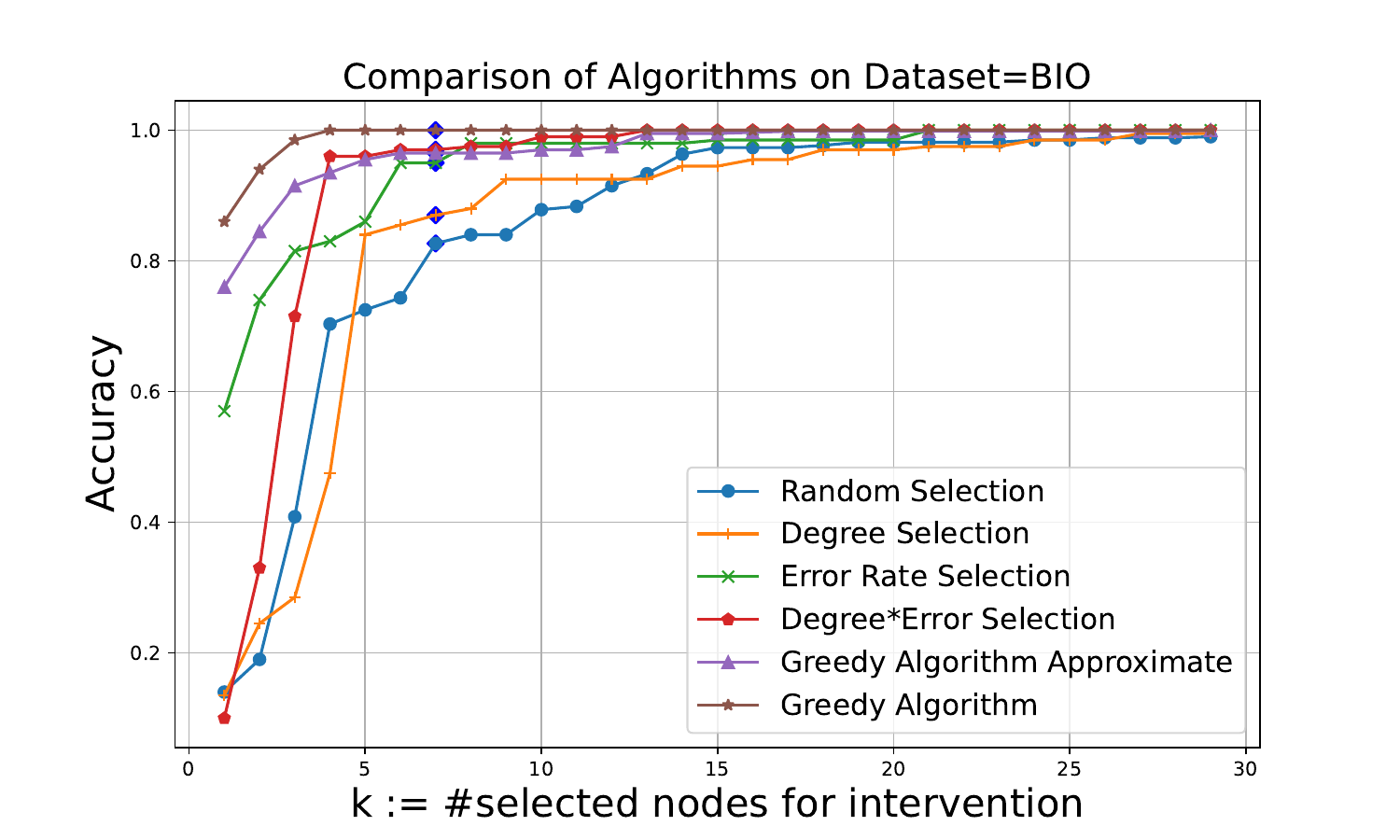} \\     \includegraphics[width=.47\linewidth]{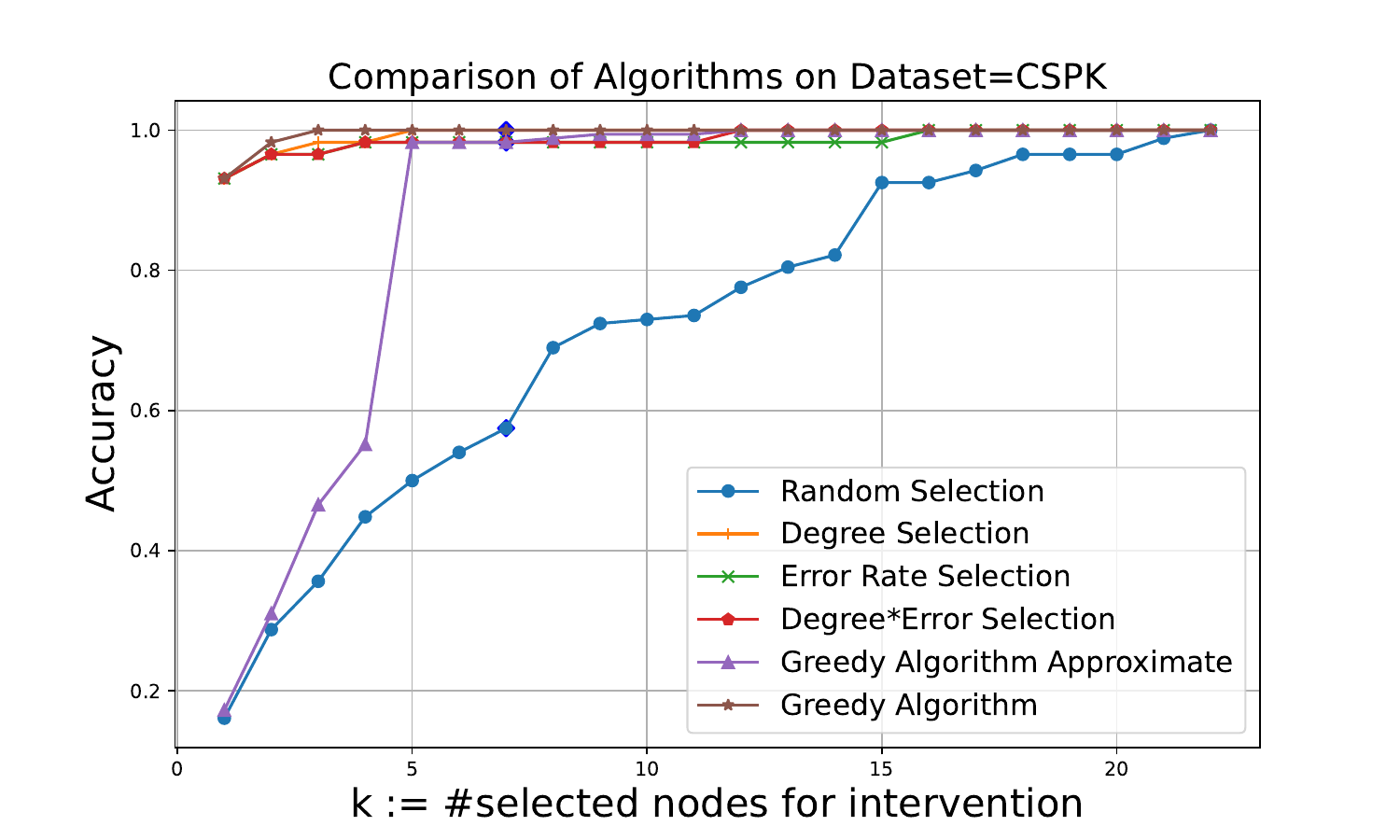} & \includegraphics[width=.47\linewidth]{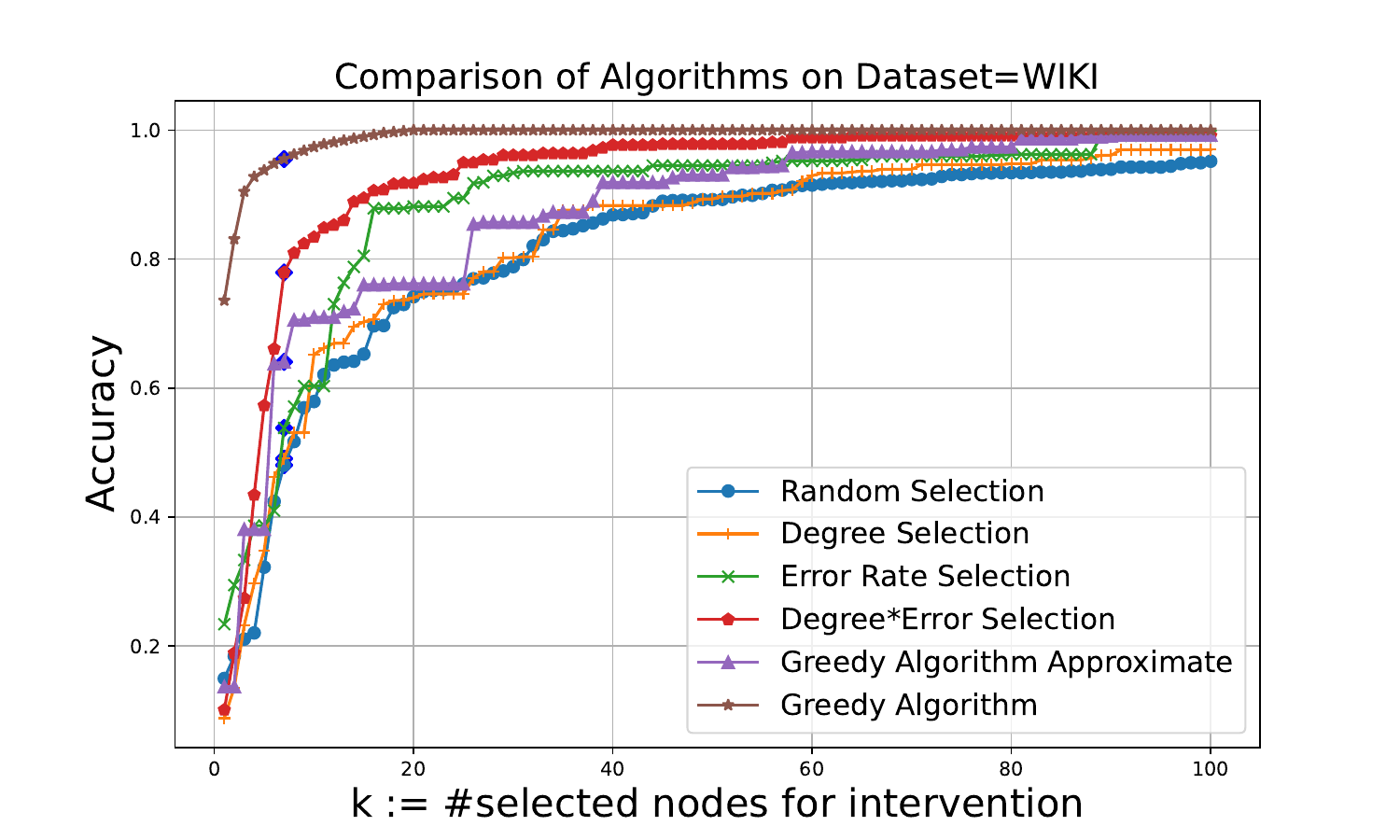}\\
    \end{tabular}
      \caption{More experimental results on different datasets.}
      \label{fig:exp_appendix}
\end{figure}



\end{document}